\documentclass[a4paper,UKenglish,cleveref, autoref, 
thm-restate,authorcolumns]{lipics-v2019}

\usepackage{graphicx} 
\usepackage{amssymb}
\usepackage{amsmath}
\usepackage{times}
\usepackage{algorithm,algorithmic}
\usepackage{enumerate}
\usepackage{todonotes}
\usepackage{hyperref}
\hypersetup{%
  colorlinks=true,           
  allcolors=blue!70!black,   
  pdfstartview=Fit,          
  breaklinks=true,
  pdfauthor={B. Bonakdarpour and S. Sheinvald},
  pdftitle={Automata for Hyperproperties}}
  
\title{Automata for Hyperlanguages} 

\titlerunning{Automata for Hyperlanguages} 

\author{Borzoo Bonakdarpour\footnote{Optional footnote, e.g. to mark 
corresponding author}}{Department of Computer Science, Iowa State University, 
U.S.A.}{borzoo@iastate.edu}{https://orcid.org/0000-0003-1800-5419}{
}

\author{Sarai Sheinvald}{Department of Software Engineering, ORT Braude College, 
Israel 
}
{sarai@braude.ac.il}{}{
}

\authorrunning{B. Bonakdarpour and S. Sheinvald} 

\Copyright{B. Bonakdarpour and S. Sheinvald} 

\begin{document}
\maketitle
\nolinenumbers

\newcommand{\f}{\varphi}
\newcommand{\g}{\psi}

\newcommand{\globally}{\textsf {G}\, }
\newcommand{\eventually}{\textsf{F}\, }
\newcommand{\weakuntil}{\textsf{W}\, }
\newcommand{\until}{\, \textsf{U}\, }
\newcommand{\releases}{\, \textsf{V}\, }
\newcommand{\nextt}{\textsf{X}\, }
\newcommand{\true}{\textbf{\textit{tt}}}
\newcommand{\false}{\textbf{\textit{ff}}}

\newcommand{\zug}[1]{\langle #1 \rangle}
\newcommand{\tuple}[1]{\langle #1 \rangle}
\newcommand{\A}{\mathcal{A}}
\newcommand{\B}{\mathcal{B}}
\newcommand{\lang}[1]{\mathcal{L}(#1)}
\newcommand{\hlang}[1]{\mathfrak{L}(#1)}
\renewcommand{\hl}{\mathfrak{L}}
\newcommand{\K}{\mathcal{K}}
\newcommand{\M}{\mathcal{M}}
\newcommand{\D}{\mathcal{D}}
\newcommand{\U}{\mathcal{~U~}}
\newcommand{\lstar}{\textsc{L}^*}
\newcommand{\ceil}[1]{\lceil #1 \rceil}
\newcommand{\nfhef}{\textrm{NFH}_{\exists\forall}}
\newcommand{\nfhfe}{\textrm{NFH}_{\forall\exists}}
\newcommand{\nfhe}{\textrm{NFH}_{\exists}}
\newcommand{\nfhf}{\textrm{NFH}_{\forall}}

\newcommand{\comp}[1]{\textsf{\small #1}}

\newcommand{\borzoo}[1]{\textcolor{red}{\bf #1}}
\newcommand{\alphabet}{\Sigma}

\newcommand{\naturals}{\mathbb{N}}
\newcommand{\N}{\naturals}

\newcommand{\bi}[1]{\textbf{\textit #1}} 


\newcommand{\quant}{\mathbb{Q}}
\newcommand{\stam}[1]{}
\newcommand{\zip}{\mathsf{zip}}
\newcommand{\unzip}{\mathsf{unzip}}
\newcommand{\row}{\mathsf{row}}

\begin{abstract}

{\em Hyperproperties} lift conventional trace properties from a set of 
execution traces to a set of sets of execution traces. Hyperproperties have 
been shown to be a powerful formalism for expressing and reasoning about 
information-flow security policies and important properties of cyber-physical 
systems such as sensitivity and robustness, as well as consistency conditions in 
distributed computing such as linearizability. Although there is an extensive 
body of work on automata-based representation of trace properties, we currently 
lack such characterization for hyperproperties.

We introduce {\em hyperautomata} for {\em hyperlanguages}, which are languages 
over sets of words. Essentially, hyperautomata allow running multiple quantified 
words over an automaton. We propose a specific type of hyperautomata called 
{\em nondeterministic finite hyperautomata} (NFH), which accept {\em regular 
hyperlanguages}. We demonstrate the ability of regular hyperlanguages to express 
hyperproperties for finite traces. We then explore the fundamental properties of 
NFH and show their closure under the Boolean operations.
We show that while nonemptiness is undecidable in general, it is decidable for several fragments of NFH. We further show the decidability of the membership problem for finite sets and regular languages for NFH, as well as the containment problem for several fragments of NFH. 
Finally, we introduce learning algorithms 
based on Angluin’s $\lstar$ algorithm for the fragments NFH in which the 
quantification is either strictly universal or strictly existential.

\end{abstract}
\newpage
\section{Introduction}

Hyperproperties~\cite{cs10} generalize the traditional trace 
properties~\cite{as85} to {\em system properties}, i.e., a set of sets of 
traces. Put it another way, a hyperproperty prescribes how the system should 
behave in its entirety and not just based on its individual executions. 
Hyperproperties have been shown to be a powerful tool for expressing and
reasoning about information-flow security policies~\cite{cs10} and important 
properties of cyber-physical systems~\cite{wzbp19} such as sensitivity and 
robustness, as well as consistency conditions in distributed computing such as
linearizability~\cite{bss18}. 

Automata theory has been in the forefront of developing techniques for 
specification and verification of computing systems. For instance, in the
automata-theoretic approach to verification~\cite{vw86,VW94}, the model-checking 
problem is reduced to checking the nonemptiness of the product automaton of the 
model and the complement of the specification. In the industry and other 
disciplines (e.g., control theory), automata are an appealing choice for 
modeling the behavior of a system. Unfortunately, we currently lack a deep 
understanding about the relation between hyperproperties and automata theory. 
To our knowledge, work in this area is limited to ~\cite{fht19}, in which 
the authors develop an automata representation for the class of regular 
$k$-safety hyperproperties. These are hyperproperties where execution traces 
are only universally quantified and their behaviors are non-refutable. They 
introduce the notion of a {\em $k$ bad-prefix automaton} -- a finite-word automaton  
that recognizes sets of $k$ bad  prefixes as finite words. Based on this 
representation, they present a learning algorithm for $k$-safety 
hyperproperties. In~\cite{frs15}, the authors offer a model-checking 
algorithm for hyperCTL$^*$~\cite{cfkmrs14}, which constructs an alternating  
B{\"u}chi automaton that has both the formula and the Kripke structure 
``built-in''. These approaches translate a hyperproperty-related problem to word 
automata.

We generalize the idea in \cite{fht19} to a broader view of an automata-based 
representation of hyperproperties, and introduce {\em hyperautomata} for {\em 
hyperlanguages}, which are languages whose elements are sets of {\em finite} 
words, which we call {\em hyperwords}. In this paper, we propose {\em 
nondeterministic finite-word hyperautomata} (NFH). An NFH runs on {\em 
hyperwords} that contain finite words, by using quantified {\em word variables} 
that range over the words in a hyperword, and a nondeterministic finite-word 
automaton (NFA) that runs on the set of words that are assigned to the 
variables. We demonstrate the idea with two examples.

\begin{example}

Consider the NFH $\A_1$ in Figure~\ref{fig:nfh_examples} (left), whose alphabet 
is $\Sigma = \{a,b\}$, over two word variables $x_1$ and $x_2$. The NFH $\A_1$ contains an underlying standard NFA,
whose alphabet comprises pairs over $\Sigma$, i.e., elements of $\Sigma^2$, in which the first letter represents the letters of the word 
assigned to $x_1$, and dually for the second letter and $x_2$.
The underlying NFA of $\A_1$ requires that (1) these two words agree on their 
$a$ (and, consequently, on their $b$) positions, and (2) once one of the words has ended 
(denoted by $\#$), the other must only contain $b$ letters. Since the 
{\em quantification condition} of $\A_1$ is $\forall x_1 \forall x_2$, in a hyperword $S$ that is accepted by $\A_1$, every two words agree on their $a$ positions. As a result, all the words in $S$ must agree on their $a$ 
positions. The hyperlanguage of $\A_1$ is then all hyperwords in which all 
words agree on their $a$ positions. 

\end{example}

\begin{example}
Next, consider the NFH $\A_2$ in Figure~\ref{fig:nfh_examples} (right), over 
the alphabet $\alphabet = \{a\}$, and two word variables $x_1$ and $x_2$. The 
underlying NFA of $\A_2$ accepts the two words assigned to $x_1$ and $x_2$ 
iff the word assigned to $x_2$ is longer than the word assigned to $x_1$. Since 
the quantification condition of $\A_2$ is $\forall x_1 \exists x_2$, we have 
that $\A_2$ requires that for every word in a hyperword $S$ accepted by $\A_2$, 
there exists a longer word in $S$. This holds iff $S$ contains infinitely many 
words. Therefore, the hyperlanguage of $\A_2$ is the set of all infinite 
hyperwords over $\{a\}$. 

\end{example}

\begin{figure}[ht]
\centering
\scalebox{.8}{
        \includegraphics[scale=0.5]{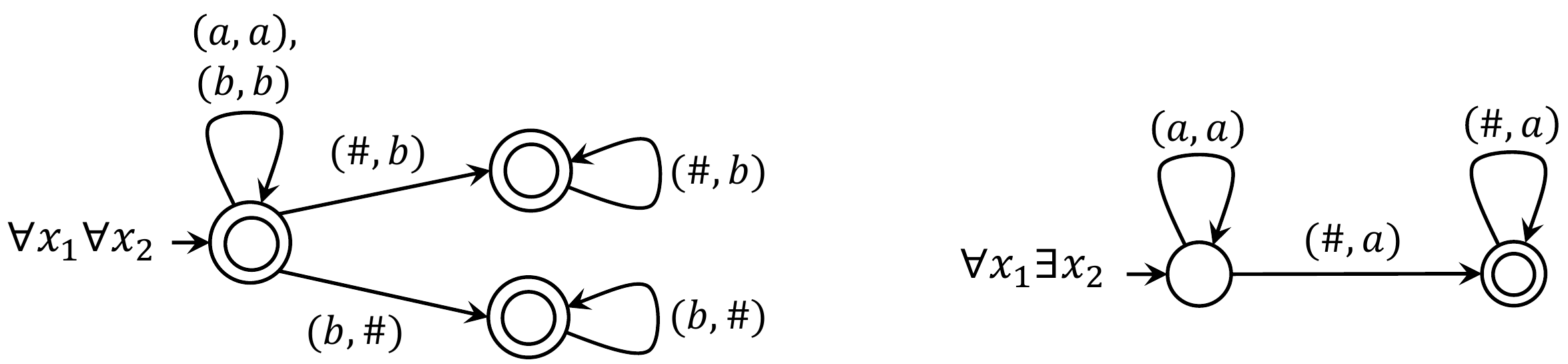}
    }
    \caption{The NFH $\A_1$ (left) and $\A_2$ (right).}
    \label{fig:nfh_examples}
\end{figure}

We call the hyperlanguages accepted by NFH {\em regular hyperlanguages}.
A regular hyperlanguage $\hl$ can also be expressed by the regular expression 
for the language of the underlying NFA of an NFH $\A$ for $\hl$, augmented with 
the quantification condition of $\A$. We call such an expression a {\em 
hyperregular expression} (HRE). We demonstrate the ability of HREs to express 
important information-flow security policies such as different variations of 
{\em noninteference}~\cite{gm82} and {\em observational 
determinism}~\cite{zm03}.
 
We proceed to conduct a comprehensive study of properties of NFH  (see 
Table~\ref{tab:results}). In particular, we show that NFH are {\em closed} 
under union, intersection, and complementation. We also prove that the 
{\em nonemptiness} 
problem is in general undecidable for NFH. However, for the alternation-free 
fragments (which only allow one type of quantifier), as well as for the 
$\exists\forall$ fragment (in which the quantification condition is limited to a 
sequence of $\exists$ quantifiers followed by a sequence of $\forall$ 
quantifiers), nonemptiness is decidable. 
These results are in line with the results on satisfiability of 
HyperLTL~\cite{fh16}. We also study the {\em membership} and {\em inclusion} 
problems. These results are aligned with the complexity of HyperLTL model 
checking for tree-shaped and general Kripke structures~\cite{bf18}. This shows 
that, surprisingly, the complexity results in~\cite{fh16,bf18} mainly stem from 
the nature of quantification over finite words and depend on neither the full 
power of the temporal operators nor the infinite nature of HyperLTL semantics.

Finally, we introduce learning algorithms for the alternation-free fragments 
of NFH. Our algorithms are based on Angluin's $\lstar$ algorithm 
\cite{Angluin87} for regular languages, and are inspired by~\cite{fht19}, where the authors describe a learning algorithm that is tailored to learn a $k$-bad prefix NFA for a $k$-safety formula. In fact, the algorithm there can be viewed of as a special case of learning a hyperlanguage in the $\exists$-fragment of NFH. 

In a learning algorithm, a {\em learner} aims to construct an automaton for an unknown target language $\cal L$, by means of querying a {\em teacher}, who knows $\cal L$. The learner asks two types of queries: {\em membership queries} (``is the word $w$ in $\cal L$?'') and {\em equivalence queries} (``is $A$ an automaton for $\cal L$?''). In case of a failed equivalence query, the teacher returns a counterexample word on which $A$ and $\cal L$ differ.  
The learning algorithm describes how the learner uses the answers it gets from the teacher to construct its candidate automaton. 

In the case of NFH, the membership queries, as well as the counterexamples, are hyperwords. The number of variables is unknown in advance, and is also part of the learning goal. 
We first define canonical forms for the alternation-free fragments of NFH, which is essential for this type of learning algorithm.
Then, we proceed to describe the learning algorithms for both fragments.

\begin{table}[t]
\begin{center}
\begin{tabular}{|c||c|c|}

\multicolumn{1}{c}{\bf Property} & \multicolumn{2}{c}{\bf Result}\\
\hline\hline
Closure &  \multicolumn{2}{c|}{Complementation, Union, Intersection 
(Theorem~\ref{thm:nfh.operations})} \\
\hline

\multirow{3}{*}{Nonemptiness}& $\forall\exists\exists$ & Undecidable 
(Theorem~\ref{thm:nfh.nonemptiness})\\ 
& $\exists^*/\forall^*$ & \comp{NL-complete} 
(Theorem~\ref{thm:nfhe.nfhf.nonemptiness})\\ 
& $\exists^*\forall^*$ & \comp{PSPACE-complete} 
(Theorem~\ref{thm:nfhef.nonemptiness})\\
\hline
\multirow{2}{*}{Finite membership}& NFH & \comp{PSPACE} 
(Theorem~\ref{thm:nfh.membership.finite})\\ 
& $O(\log(k))$ ~$\forall$ & \comp{NP-complete} 
(Theorem~\ref{thm:nfh.membership.finite})\\ 
\hline
Regular membership & \multicolumn{2}{c|}{Decidable 
(Theorem~\ref{thrm:membershipFULL})} \\
\hline
Containment & $\exists^*/\forall^*/\exists^*\forall^*\subseteq \exists^*/\forall^*$ & \comp{PSPACE-complete} 
(Theorem~\ref{thrm:containment})\\
\hline

\end{tabular} 
\caption{Summary of results on properties of NHF.}
\label{tab:results}
\end{center}
\end{table}

\noindent {\em Organization.} \ The rest of the paper is organized as follows. 
Preliminary concepts are presented in Section~\ref{sec:prelim}. We introduce 
the notion of NFH and HRE in Sections~\ref{sec:ha} and~\ref{sec:hre}, while 
their properties are studied in Section~\ref{sec:nfh_properties}. We propose 
our learning algorithm in Section~\ref{sec:learning}. Finally, we make 
concluding remarks and discuss future work in Section~\ref{sec:concl}. 
Detailed proofs appear in the appendix.
\section{Preliminaries}
\label{sec:prelim}

An {\em alphabet} is a nonempty finite set $\Sigma$ of {\em letters}. 
A {\em word} over $\Sigma$ is a finite 
sequence of letters from 
$\Sigma$. The {\em empty word} is denoted by $\epsilon$, and the set of all finite words is denoted by $\Sigma^*$. A {\em language} is a subset of $\Sigma^*$.

\begin{definition}
\label{def:nfa}
A {\em nondeterministic finite-word automaton} (NFA) is a tuple \linebreak $A 
= \tuple{\Sigma,Q,Q_0,\delta,F}$, where $\Sigma$ is an alphabet, $Q$ is a 
nonempty finite set of {\em states}, $Q_0\subseteq Q$ is a set of {\em initial 
states}, $F\subseteq Q$ is a set of {\em accepting states}, and 
$\delta\subseteq Q\times\Sigma\times Q$ is a {\em transition relation}. 
\end{definition}

Given a word $w=\sigma_1\sigma_2\cdots \sigma_n$ over $\Sigma$, a 
{\em run of $A$ on $w$} is a sequence of states $(q_0,q_1,\ldots q_n)$, such 
that $q_0\in Q_0$, and for every $0 < i \leq n$, it holds that 
$(q_{i-1},\sigma_i, q_i)\in \delta$.
The run is {\em accepting} if $q_n\in F$. 
We say that $A$ {\em accepts} $w$ if there exists an accepting run of $A$ on $w$. 
The {\em language} of $A$, denoted by $\lang{A}$, is the set of all finite words that $A$ accepts. 
A language $\cal L$ is called {\em regular} if there exists an NFA such that $\lang{A} = \cal L$.

An NFA $A$ is called {\em deterministic} (DFA), if for every $q\in Q$ and 
$\sigma\in\Sigma$, there exists exactly one $q'$ for which $(q,\sigma,q')\in 
\delta$, i.e., $\delta$ is a transition {\em function}.
It is well-known that every NFA has an equivalent DFA.

\stam{
A {\em regular expression} over $\Sigma$ is defined inductively as follows. $\epsilon, \emptyset$, and $\sigma\in \Sigma$ are regular expressions. Let $r_1,r_2$ be regular expressions. Then $(r_1|r_2)$, $(r_1\cdot r_2)$, and $(r_1^*)$ are regular expressions. 
The semantics of regular expressions assigns every regular expression $r$ a language $\lang{r}$, as follows. $\lang{\epsilon} = \{\epsilon\}$, $\lang{\emptyset} = \emptyset$, and $\lang{\sigma} = \{\sigma\}$. Additionally, $\lang{(r_1+r_2)}= \lang{r_1}\cup \lang{r_2}$, $\lang{(r_1\cdot r_2)} = \lang{r_1}\cdot \lang{r_2}$, and $\lang{(r_1^*)} = \lang{r_1}^*$. It is well-known that $\cal L$ is regular iff there exists a regular expression $r$ such that $\lang{r} = \cal L$. 
}

\stam{
\begin{definition}
 \label{def:nba}
A {\em nondeterministic B{\"u}chi automaton} (NBA) is a tuple \linebreak $A 
= \tuple{\Sigma,Q,Q_0,\delta,F}$, whose elements are defined as in 
Definition~\ref{def:nfa}. However, the runs of $A$ are defined over infinite 
words. An infinite word $w$ is accepted by $A$ if $A$ has an infinite run on $w$ 
that visits a state in $F$ infinitely often.
\end{definition}

As before, the language $\lang{A}$ is the set of all infinite words accepted by 
$A$. }

\section{Hyperautomata}
\label{sec:ha}


Before defining hyperautomata, we explain the idea behind them. We first define hyperwords and hyperlanguages.

\begin{definition}
\label{def:hword}
A {\em hyperword over $\Sigma$} is a set of words over $\Sigma$ and a 
{\em hyperlanguage} is a set of hyperwords.
\end{definition}

A 
hyperautomaton $\A$ uses a set of {\em word variables} $X  =\{x_1,x_2,\ldots, x_k\}$. 
When running on a hyperword $S$, these variables are assigned words from $S$. We 
represent an assignment $v:X\rightarrow S$ as the $k$-tuple 
$(v(x_1),v(x_2),\ldots, v(x_k))$. Notice that the variables themselves do not 
appear in this representation of $v$, and are manifested in the order of the 
words in the $k$-tuple: the $i$'th word is the one assigned to $x_i$. This 
allows a cleaner representation with less notations. 

The hyperautomaton $\A$ consists of a {\em quantification condition} $\alpha$ 
over $X$, and an {\em underlying word automaton} $\hat\A$, which runs on words 
that represent assignments to $X$ (we explain how we represent assignments as 
words later on). The condition $\alpha$ defines the assignments that $\hat\A$ 
should accept. For example, $\alpha = \exists x_1\forall x_2$ requires that 
there exists a word $w_1\in S$ (assigned to $x_1$), such that for every word 
$w_2\in S$ (assigned to $x_2$), the word that represents $(w_1,w_2)$ is accepted 
by $\hat\A$. The hyperword $S$ is accepted by $\A$ iff $S$ meets these 
conditions. 

We now elaborate on how we represent an assignment $v:X\rightarrow S$ as a word. 
We encode the tuple $(v(x_1),v(x_2),\ldots v(x_k))$ by a word ${\bi w}$ whose 
letters are $k$-tuples in $\alphabet^k$, where the $i$'th 
letter of ${\bi w}$ represents the $k$ $i$'th letters of the words 
$v(x_1),\ldots ,v(x_k)$ (in case that the words are not of equal length, we ``pad'' 
the end of the word with $\#$ signs). 
For example, the assignment $v(x_1)=aa,v(x_2)=abb$, represented by the tuple $(aa,abb)$, is encoded by the word 
$ (a,a)(a,b)(\#,b)$.
We later refer to ${\bi w}$ as the {\em zipping} of $v$. Once again, notice that due to the indexing of the word variables, the variables do not explicitly appear in ${\bi w}$.   


We now turn to formally define  hyperautomata.

\subsection{Nondeterminsitic Finite-Word Hyperautomata}
\label{sec:haf}

We begin with some terms and notations. 

Let $s = (w_1,w_2,\ldots, w_k)$ be a tuple of finite words over  $\Sigma$. We 
denote the length of the longest word in $s$ by $\ceil{s}$. We represent $s$ by 
a word over $(\Sigma\cup\{\#\})^k$ of length $\ceil{s}$, which is 
formed by a function $\zip(s)$ that ``zips'' the words in $s$ together: the 
$i$'th letter in $\zip(s)$ represents the $i$'th letters in $w_1, w_2, \ldots, 
w_k$, 
and $\#$ is used to pad the words that have ended. For example,
$$\zip(aab, bc, abdd) = (a, b, a)(a, c, b)(b, \#, d)(\#, \#, d).$$
Formally, we have $\zip(s) =  \bi{s}_1\bi{s}_2\cdots \bi{s}_{\ceil{s}}$, 
where
$\bi{s}_i[j] = w_{j_i}$ if $j\leq|w|$, and $\bi{s}_i[j] = \#$, otherwise.

Given a zipped word $\bi{s}$, we denote the word formed by the letters in the 
$i$'th positions in $\bi{s}$ by $\bi{s}[i]$. 
That is, $\bi{s}[i]$ is the word $\sigma_1\sigma_2\cdots \sigma_m$ formed by 
defining $\sigma_j = \bi{s}_j[i]$, for $\bi{s}_j[i]\in \Sigma$.
Notice that $\zip(s)$ is reversible, and we can define an $\unzip$ function 
as $\unzip(\bi{s}) = (\bi{s}[1],\bi{s}[2],\dots, \bi{s}[k])$. We sometimes abuse the notation, and use $\unzip(\bi{s})$ to denote $\{\bi{s}[1],\bi{s}[2],\dots, \bi{s}[k]\}$, and $\zip(S)$ to denote the zipping of the words in a finite hyperword $S$ in some arbitrary order.

\begin{definition}
\label{def:nfh}
A {\em nondeterministic finite-word hyperautomaton} (NFH) is a tuple \linebreak
$\A = \tuple{\Sigma,X,Q,Q_0,F,\delta,\alpha}$, where $\Sigma$, $Q$, $Q_0$, and $F$ are as in Definition~\ref{def:nfa}, $X=\{x_1,\dots, x_k\}$ is a finite set of {\em word variables}, 
$\delta\subseteq Q\times (\Sigma\cup \{\#\})^k \times Q$ is a transition 
relation, and $\alpha  = \quant_1 x_1\quant_2x_2\ldots \quant_nx_k$ is a {\em 
quantification condition}, where $\quant_i\in\{\forall,\exists\}$ for every 
$1\leq i\leq k$.

\end{definition}
In Definition~\ref{def:nfh}, the tuple $\tuple{(\Sigma\cup \{\#\})^k, Q, 
Q_0, \delta, F}$ forms an underlying NFA of $\A$, which we denote by 
$\hat{\A}$. We denote the alphabet of $\hat\A$ by $\hat{\Sigma}$. 

Let $S$ be a hyperword and let $v: X\rightarrow S$ be an assignment of the word 
variables of $\A$ to words in $S$. We denote by  $v[x\rightarrow w]$ the assignment obtained from $v$ by assigning the word $w\in S$ to $x\in X$. We represent $v$ by the word $\zip(v) = \zip(v(x_1),\ldots v(x_k))$.  
We now define the acceptance condition of a hyperword $S$ by an NFH $\A$. We 
first define the satisfaction relation $\models$ for $S$, $\A$, a quantification condition $\alpha$, and an assignment $v:X\rightarrow S$, as follows.

\begin{itemize}
    \item For $\alpha = \epsilon$, we denote $S \models _v (\alpha,\A)$ if
$\hat\A$ accepts $\zip(v)$. 

\item For $\alpha = \exists x_i \alpha'$, we denote $S\models_v (\alpha,\A)$ if 
there exists $w\in S$, such that $S \models_{v[x_i\rightarrow w]}  
(\alpha',\A)$.

\item For $\alpha = \forall x_i \alpha'$, we denote $S\models_v (\alpha,\A)$ 
if for every $w\in S$, it holds that $S \models_{v[x_i\rightarrow w]}  
(\alpha',\A)$.\footnote{In case that $\alpha$ begins with 
$\forall$, satisfaction holds vacuously with an empty hyperword. We 
restrict the discussion to nonempty hyperwords.}

\end{itemize}
Since the quantification condition of $\A$ includes all of $X$, the satisfaction is independent of the 
assignment $v$, and we denote $S \models \A$, in which case, we say that {\em $\A$ 
accepts $S$}.

\begin{definition}
Let $\A$ be an NFH. The {\em hyperlanguage} of $\A$, denoted $\hlang{\A}$, is 
the set of all hyperwords that $\A$ accepts.
\end{definition}

We call a hyperlanguage $\hl$ a {\em regular hyperlanguage} if there exists an NFH $\A$ such that $\hlang{\A} = \hl$.


\begin{example}
Consider the NFH $\A_3$ in Figure~\ref{fig:ordered}, over the alphabet $\Sigma = 
\{a,b\}$ and two word variables $x_1$ and $x_2$. From the initial state, two 
words lead to the left component in $\hat{\A_3}$ iff in every position, if the 
word assigned to $x_2$ has an $a$, the word assigned to $x_1$ has an $a$. In 
the right component, the situation is dual -- in every position, if the word 
assigned to $x_1$ has an $a$, the word assigned to $x_2$ has an $a$. 
Since the quantification condition of $\A_3$ is $\forall x_1\forall x_2$, in a hyperword $S$ accepted by $\A_3$, in every two words in $S$, the set of $a$ positions of one is a subset of the $a$ positions of the other. Therefore, $\hlang{\A_3}$ includes all hyperwords in which there is a full ordering on the $a$ positions.

\begin{figure}[ht]
    \begin{center}
        \includegraphics[scale=0.5]{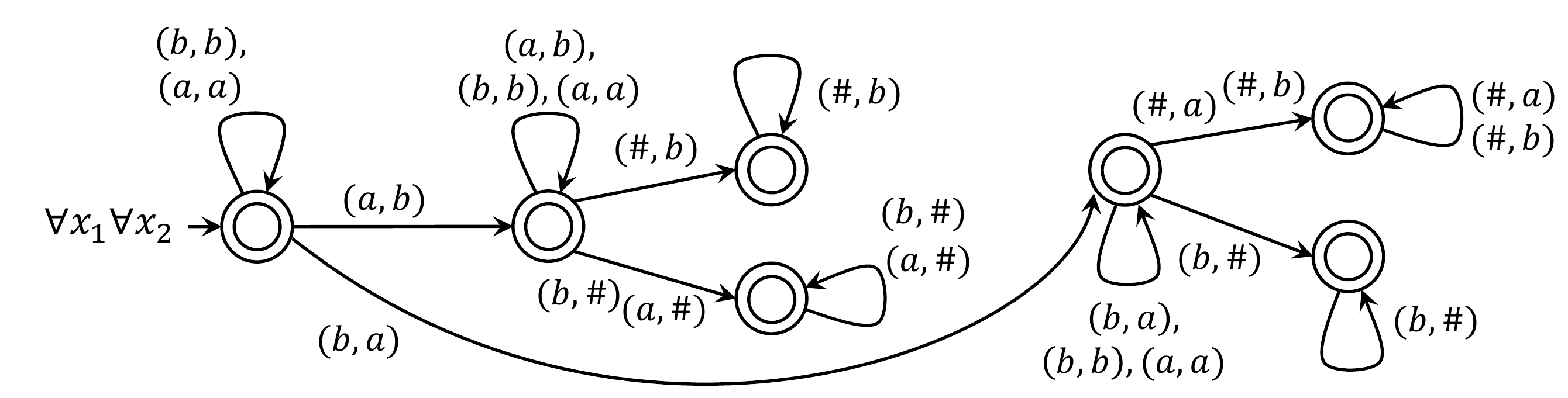}
    \end{center}
    \caption{The NFH $\A_3$.}
    \label{fig:ordered}
\end{figure}
\end{example}

We consider several fragments of NFH, which limit the structure of the quantification condition $\alpha$.
$\nfhf$ is the fragment in which $\alpha$ contains only $\forall$ quantifiers, 
and similarly, in $\nfhe$, $\alpha$ contains only $\exists$ quantifiers. In 
the fragment $\nfhef$, $\alpha$ is of the form $\exists x_1 \cdots \exists x_i \forall 
x_{i+1}\cdots \forall x_k$.


\subsection{Additional Terms and Notations}

We present several more terms and notations which we use throughout the following sections.
We say that a word $\bi w$ over $(\Sigma\cup \#)^k$ is {\em legal} if 
${\bi w}=\zip(u_1,\ldots u_k)$ for some $u_1,u_2,\ldots u_k \in \Sigma^*$. 
Note that $\bi w$ is legal iff there is no ${\bi w}[i]$ in which there is an occurrence of $\#$ 
followed by some letter $\sigma\in \Sigma$. 

Consider two letter tuples $\sigma_1 = (t_1,\ldots t_k)$ and $\sigma_2 = (s_1,\ldots s_{k'})$. We denote by $\sigma_1+\sigma_2$ the tuple $(t_1,\ldots t_k, s_1,\ldots s_{k'})$.
We extend the notion to zipped words. Let ${\bi w_1} = \zip(u_1,\ldots u_k)$ and ${\bi w_2} = \zip(v_1,\ldots v_{k'})$. We denote by ${\bi w_1} +{\bi w_2}$ the word $\zip(u_1,\ldots u_k,v_1,\ldots v_{k'})$.

Consider a tuple $t = (t_1,t_2,\ldots t_k)$ of items. 
A {\em sequence} of $t$ is a tuple $(t'_1, t'_2,\ldots t'_k)$, where 
$t'_i\in\{t_1,\ldots t_k\}$ for every $1\leq i \leq k$. A {\em permutation} of $t$ is a reordering of the elements of $t$. 
We extend these notions to zipped words, to 
assignments, and to hyperwords, as follows. Let $\zeta = 
(i_1,i_2,\ldots i_k)$ be a sequence (permutation) of $(1,2,\ldots, k)$.
\begin{itemize}
    \item Let ${\bi w} = \zip(w_1,\ldots w_k)$ be a word over $k$-tuples. The word ${\bi w}_\zeta$, 
defined as $\zip(w_{i_1}, w_{i_2}, \ldots w_{i_k})$ is a sequence (permutation) of ${\bi w}$.
    \item Let $v$ be an assignment from a set of variables $\{x_1,x_2,\ldots 
x_k\}$ to a hyperword $S$. The assignment $v_\zeta$, defined as $v_\zeta(x_j) = 
v(x_{i_j})$ for every $1\leq i,j \leq k$, is a sequence (permutation) of $v$. 
\item Let $S$ be a hyperword. The tuple ${\bi w} = (w_1,\ldots w_k)$, where $w_i\in S$, is a sequence of $S$. if $\{w_1,\ldots w_k\} = S$, then $\bi w$ is a permutation of $S$. 
\end{itemize}

\section{Hyperregular Expressions and Application in Security}
\label{sec:hre}

Given an NFH $\A$, the language of its underlying NFA $\hat\A$ can be expressed as a regular expression $r$. Augmenting $r$ with the 
quantification condition $\alpha$ of $\A$ constitutes a {\em hyperregular expression} (HRE) $\alpha r$. For example, consider the NFH $\A_1$ in Figure~\ref{fig:nfh_examples}. The HRE of $\A_1$ is:
$$
\forall x_1\forall x_2\Big((a, a) \mid (b, b)\Big)^*\Big((\#, b)^* \mid (b, 
\#)^* \Big)
$$
We now show the application of HREs in specifying well-known information-flow 
security policies.

{\em Noninteference}~\cite{gm82} requires that commands issued by users holding 
high clearances be removable without affecting observations of users holding 
low clearances:
$$
\varphi_{\mathsf{ni}} = \forall x_1\exists x_2(l, l\lambda)^*
$$
where $l$ denotes a low state and $l\lambda$ denotes a low state where all high 
commands are replaced by a dummy value $\lambda$.

{\em Observational determinism}~\cite{zm03} requires that if two executions of 
a system start with low-security-equivalent events, then these 
executions should remain low equivalent:
$$
\varphi_{\mathsf{od}} = \forall x_1\forall x_2 (l, l)^+ \mid (\bar{l}, 
\bar{l})(\$, \$)^* \mid (l, \bar{l})(\$, \$)^* \mid (\bar{l}, l)(\$, \$)^*
$$
where $l$ denotes a low event, $\bar{l} \in \Sigma \setminus \{l\}$, and $\$ 
\in 
\Sigma$. We note that similar policies such as {\em Boudol and Castellani’s 
noninterference}~\cite{bd02} can be formulated in the same 
fashion.\footnote{This policy states that every two executions that start from 
bisimilar states (in terms of memory low-observability), should remain 
bisimilarly low-observable.}

{\em Generalized noninterference} (GNI)~\cite{m88} allows nondeterminism in 
the low-observable behavior, but requires that low-security outputs may 
not be altered by the injection of high-security inputs:
$$
\varphi_{\mathsf{gni}} = \forall x_1\forall x_2\exists x_3 \bigg((h, l, hl) 
\mid (\bar{h}, l, \bar{h}l) \mid (h, \bar{l}, h\bar{l}) \mid (\bar{h}, \bar{l}, 
\bar{h}\bar{l}) \bigg)^*
$$
where $h$ denotes the high-security input, $l$ denotes the low-security output, 
$\bar{l} \in \Sigma \setminus\{l\}$, and $\bar{h} \in \Sigma \setminus \{h\}$.

{\em Declassification}~\cite{ss00} relaxes noninterference by allowing leaking 
information when necessary. Some programs need to reveal secret information to 
fulfill functional requirements. For example, a password checker must reveal 
whether the entered password is correct or not:
$$
\varphi_{\mathsf{dc}} = \forall x_1\forall x_2 (li,li)(pw, pw)(lo, lo)^+
$$
where $li$ denotes low-input state, $pw$ denotes that the password is correct, 
and $lo$ denotes low-output states. We note that for brevity, in the above 
formula, we do not include behaviors where the first two events are not low or 
in the second event, the password is not valid. 

{\em Termination-sensitive noninterference} requires that for two executions 
that
 start from low-observable states, information leaks are not permitted by the termination behavior of the program:
$$
\varphi_{\mathsf{tsni}} = \forall x_1\forall x_2 (l, l)(\$, \$)^*(l, l)  \mid 
(\bar{l}, \bar{l})(\$, \$)^* \mid (l, \bar{l})(\$, \$)^* \mid (\bar{l}, l)(\$, 
\$)^*
$$
where $l$ denotes a low state and $\$ \in \Sigma$.

\section{Properties of Regular Hyperlanguages}
\label{sec:nfh_properties}

In this section, we consider the basic operations and decision 
problems for the various fragments of NFH. We mostly provide proof sketches, and 
the complete details appear in the appendix.
Throughout this section, $\A$ is an NFH 
$\tuple{\Sigma,X,Q,Q_0,\delta,F,\alpha}$, where $X = \{x_1,\ldots x_k\}$. 


We first show that NFH are closed under all the Boolean operations.

\begin{theorem}\label{thm:nfh.operations}
NFH are closed under union, intersection, and complementation.
\end{theorem}

\begin{proof}[Proof Sketch]
Complementing $\A$ amounts to dualizing its quantification condition (replacing every $\exists$ with $\forall$ and vice versa), and complementing $\hat\A$ via the standard construction for NFA.

Now, let $\A_1$ and $\A_2$ be two NFH.
The NFH $\A_{\cap}$ for $\hlang{\A_1}\cap \hlang{\A_2}$ is based on the product construction of $\hat\A_1$ and $\hat\A_2$.
The quantification condition of $\A_{\cap}$ is $\alpha_1 \cdot \alpha_2$. 
The underlying automaton $\hat\A_{\cap}$ advances simultaneously on both $\A_1$ 
and $\A_2$: when $\hat\A_1$ and $\hat\A_2$ run on zipped hyperwords ${\bi w_1}$ 
and ${\bi w_2}$, respectively, 
$\hat\A_{\cap}$ runs on ${\bi w_1}+{\bi w_2}$, and accepts only if both 
$\hat\A_1$ and $\hat\A_2$ accept.

Similarly, the NFH $\A_{\cup}$ for $\hlang{\A_1}\cup \hlang{\A_2}$ is based on the union construction of $\hat\A_1$ and $\hat\A_2$.
The quantification condition of $\A_{\cup}$ is again $\alpha_1\cdot \alpha_2$. 
The underlying automaton $\hat\A_{\cup}$ advances either on $\A_1$ or $\A_2$. For every word $\bi w$ read by $\hat\A_1$, the NFH $\hat\A_{\cup}$ reads ${\bi w}+{\bi w'}$, for every ${\bi w}'\in \hat\Sigma_2^*$, and dually, for every word $\bi w$ read by $\hat\A_2$, the NFH $\hat\A_{\cup}$ reads ${\bi w'}+{\bi w}$, for every ${\bi w}'\in \hat\Sigma_1^*$.
\end{proof}

\stam{
\begin{theorem} 
\label{thm:nfh.complement}

NFH are closed under complementation.

\end{theorem}

\begin{theorem} \label{thm:nfh.union}
NFH are closed under union.
\end{theorem}

\begin{theorem}\label{thm:nfh.intersection}

NFH are closed under intersection.
\end{theorem}
}

We now turn to study various decision problems for NFH. We begin with the 
nonemptiness problem: given an NFH $\A$, is $\hlang{\A} = \emptyset$? We show 
that while the problem is in general undecidable for NFH, it is decidable for 
the fragments that we consider. 

\begin{theorem}
\label{thm:nfh.nonemptiness}
The nonemptiness problem for NHF is undecidable.
\end{theorem}

The proof of Theorem~\ref{thm:nfh.nonemptiness} mimics the ideas in \cite{fh16}, 
which uses a reduction from the {\em 
Post correspondence problem (PCP)} to prove the undecidability of HyperLTL 
satisfiability.

For the alternation-free fragments, we can show that a simple reachability test on their underlying automata suffices to verify nonemptiness. Hence, we have the following.

\begin{theorem} \label{thm:nfhe.nfhf.nonemptiness}
The nonemptiness problem for $\nfhe$ and $\nfhf$ is \comp{NL-complete}.
\end{theorem}

The nonemptiness of $\nfhef$ is harder, and reachability does not suffice. 
However, we show that the problem is decidable.

\stam{
\begin{lemma}\label{lemma:nfhef.nonempty}
Let $\A$ be an $\nfhef$ with a quantification condition $\alpha = \exists x_1,\ldots \exists x_m \forall 
x_{m+1}\ldots \forall x_k$, where $1 \leq m < k$. 
Then $\A$ is nonempty iff $\A$ accepts a hyperword of size at most $m$.
\end{lemma}

We now use Lemma~\ref{lemma:nfhef.nonempty} to describe a decision procedure for 
the nonemptiness of $\nfhef$. 

\begin{theorem}\label{thm:nfhef.nonemptiness}
The nonemptiness problem for an $\nfhef$ $\A$ can be decided in space that is 
polynomial in the size of $\hat\A$.
\end{theorem}
}

\begin{theorem}
\label{thm:nfhef.nonemptiness}
The nonemptiness problem for $\nfhef$ is  \comp{PSPACE-complete}.
\end{theorem}

\begin{proof}[Proof Sketch]
We can show that an $\nfhef$ $\A$ is nonempty iff it accepts a hyperword $S$ of size that is bounded by the number $m$ of $\exists$ quantifiers in $\alpha$.
We can then construct an NFA $A$ whose language is nonempty iff it accepts 
$\zip(S)$ for such a hyperword $S$. The size of $A$ is 
$O(|\delta|)^{m^{k-m}})$. Unless $\A$ only accepts 
hyperwords of size $1$, which can be easily checked, $|\delta|$ must be 
exponential in the number $k-m$ of $\forall$ quantifiers, to account for all the 
assignments to the variables under $\forall$, and so overall $|A|$ is of size 
$O(|\A|^k)$. The problem can then be decided in \comp{PSPACE} by traversing 
$A$ on-the-fly. We show that a similar result holds for the case that $k-m$ is 
fixed.

We use a reduction from the unary version of the tiling problem to prove \comp{PSPACE} lower bounds both for the general case and for the case of a fixed number of $\forall$ quantifiers. 
\end{proof}

We turn to study the membership problem for NFH: given an NFH $\A$ and a 
hyperword $S$, is $S\in\hlang{\A}$? 
When $S$ is finite, the set of possible assignments from $X$ to $S$ is finite, 
and so the problem is decidable. We call this case the {\em finite membership 
problem}. 

\begin{theorem}\label{thm:nfh.membership.finite}
\begin{itemize}
\item The finite membership problem for NFH is in \comp{PSPACE}. 
\item
The finite membership problem for NFH with $O(\log(k))$ $\forall$ quantifiers is \comp{NP-complete}. 
\end{itemize}
\end{theorem}

\begin{proof}[Proof Sketch]
We can decide the membership of a hyperword $S$ in $\hlang{\A}$ by iterating over all relevant assignments from $X$ to $S$, and for every such assignment $v$, checking on-the-fly whether $\zip(v)$ is accepted by $\hat\A$. 
This algorithm uses space of size that is polynomial in $k$ and logarithmic in $|\A|$ and in $|S|$. 

When the number of $\forall$ quantifiers in $\A$ is 
$|O(\log(k))|$, we can iterate over all assignments to the $\forall$ variables in polynomial time, while guessing assignments to the variables under $\exists$. Thus, membership in this case is in \comp{NP}.

We use a reduction from the Hamiltonian cycle problem to prove \comp{NP-hardness} for this case. Given a graph $G=\tuple{\{v_1,\ldots v_n\}, E}$, we construct a hyperword $S$ with $n$ different words of length $n$ over $\{0,1\}$, each of which contains a single $1$. We also construct an $\nfhe$ $\A$ over $\{0,1\}$ with $n$ variables, a graph construction similar to that of $G$, and a single accepting and initial state $v_1$. From vertex $v_i$ there are transitions to all its neighbors, labeled by the letter $(0)^{i-1}+(1)+(0)^{n-i}$. Thus, $\A$ accepts $S$ iff there exists an assignment $f:X\rightarrow S$ such that  $\zip(f)\in\lang{\hat\A}$. Such an assignment $f$ describes a cycle in $G$, where $f(x_i)=w_j$ matches traversing $v_i$ in the $j$'th step. The words in $S$ ensure a single visit in every state, and their length ensures a cycle of length $n$.

\noindent{\it Note:} for every hyperword of size at least $2$, the number of transitions in $\delta$ must be exponential in the number $k'$ of $\forall$ quantifiers, to account for all the different assignments to these variables. Thus, if $k = O(k')$, an algorithm that uses a space of size $k$ is in fact logarithmic in the size of $\A$. 
\end{proof}

\stam{

\begin{theorem}\label{thm:nfh.membership.finite}
Let $\A$ be an NFH and $S$ be a finite hyperword over $\Sigma$. Then it can be 
decided whether $S\in \hlang{\A}$ in space that is polynomial in $k$, and 
logarithmic in $|\hat\A|, |S|$.
\end{theorem}

\begin{proof}
As discussed in the proof of Theorem~\ref{thm:nfhef.nonemptiness}, the size of 
$\hat\Sigma$ must be exponential in the number of $\forall$ quantifiers in 
$\alpha$, and therefore is exponential in $k$. We can decide the membership of 
$S$ in $\hlang{\A}$ by iterating over all assignments of the type $X\rightarrow 
S$. For every such assignment $v$, we construct $\zip(v)$ and run $\hat\A$ on 
$\zip(v)$, on-the-fly. 
\end{proof}

The exponential size of $\hat\A$ is derived from the number of $\forall$ 
quantifiers. When the number of $\forall$ quantifiers is fixed, then 
$\hat\Sigma$ is not necessarily exponential in $k$. 

\begin{theorem}
\label{thrm:membershipA}
The finite membership problem for NFH with a fixed number of $\forall$ 
quantifiers is \comp{NP-complete}.
\end{theorem}
 
}

When $S$ is infinite, it may still be finitely represented. 
We now address the problem of deciding whether a regular language $\cal L$ 
(given as an NFA) is accepted by an NFH. We call this {\em the regular 
membership problem for NFH}. We show that this problem is decidable for the 
entire class of NFH.

\begin{theorem}
\label{thrm:membershipFULL}
The regular membership problem for NFH is decidable.
\end{theorem}

\begin{proof}[Proof Sketch]
Let $A$ be an NFA, and let $\A$ be an NFH, both over $\Sigma$.
We describe a recursive procedure for deciding 
whether $\lang{A}\in\hlang{\A}$.

For the base case of $k=1$, if $\alpha = \exists x_1$, then 
$\lang{A}\in\hlang{\A}$ iff $\lang{A}\cap \lang{\hat{\A}} \neq \emptyset$.
Otherwise, if $\alpha = \forall x_1$, then $\lang{A}\in\hlang{\A}$ iff 
$\lang{A}\notin \hlang{\overline{\A}}$, where $\overline{\A}$ is the NFH for 
$\overline{\hlang{\A}}$. 
The quantification condition for $\overline{\A}$
is $\exists x_1$, which conforms to the previous case.

For $k>1$, we construct a sequence of NFH $\A_1, \A_2, \ldots, \A_k$. 
If $\alpha$ starts with $\exists$, then we set $\A_1 = \A$. Otherwise, we set $\A_1 = \overline{\A}$.
Given $\A_i$ with a quantification condition $\alpha_i$, we construct $\A_{i+1}$ as follows. 
If $\alpha_i$ starts with $\exists$, then
the set of variables of $\A_{i+1}$ is $\{x_{i+1},\ldots x_k\}$, and the 
quantification condition $\alpha_{i+1}$ is $\quant_{i+1}x_{i+1}\cdots 
\quant_kx_k$, where $\alpha_i = \quant_ix_i \quant_{i+1}\cdots \quant_kx_k$. 
The NFH $\A_{i+1}$ is roughly constructed as the intersection between $A$ and $\hat\A_{i}$, based on the first position in every $(k-i)$-tuple letter in $\hat\Sigma_i$.
Then, $\hat\A_{i+1}$ accepts a word $\zip(u_1,\ldots u_{k-i})$ iff there 
exists a word $u\in \lang{A}$, such that $\hat\A_{i}$ accepts 
$\zip(u,u_1,\ldots u_{k-i})$.
Notice that this exactly conforms to the $\exists$ condition. 
Therefore, if $\quant_{i} = \exists$, then $\lang{A}\in\hlang{\A_i}$ iff $\lang{A}\in\hlang{\A_{i+1}}$. 

\stam{
The set of states of $\A_{i+1}$ is $ Q_i\times P'$, and the set of initial states is $Q_i^0\times P_0$. The set of accepting states is ${\cal F}_i\times  F'$. For every 
$(q\xrightarrow{\sigma_i,\ldots,\sigma_k}q')\in\delta_i$ and every 
$(p\xrightarrow{\sigma_i}p')\in \rho$, we have 
$((q,p)\xrightarrow{\sigma_{i+1},\ldots \sigma_k}(q',p'))\in\delta_{i+1}$. 

Then, $\hat\A_{i+1}$ accepts a word $\zip(u_1,u_2,\ldots u_{k-i})$ iff there 
exists a word $u\in \lang{A}$, such that $\hat\A_{i}$ accepts 
$\zip(u,u_1,u_2,\ldots u_{k-i})$. 

We first consider the case that $\quant_i = \exists$. 
Let $v:\{x_{i},\ldots x_k\}\rightarrow \lang{A}$.
Then $\lang{A}\models _v \alpha_i\A_i$ iff there exists $w\in \lang{A}$ such 
that $\lang{A}\models_{v[x_i\rightarrow w]} \alpha_{i+1},A_i$.
For an assignment $v':\{x_{i+1},\ldots x_k\}\rightarrow \lang{A}$, it holds 
that 
$\zip(v')$ is accepted by $\hat{A}_{i+1}$ iff there exists a word $w\in 
\lang{A}$ such that $\zip(v)\in\lang{\hat{A}_i}$, where $v$ is obtained from 
$v'$ 
by setting $v(x_i) = w$. 

Therefore, we have that $\lang{A}\models_{v[x_i\rightarrow 
w]}\alpha_{i+1},\A_i$ 
iff $\lang{A}\models_{v'} \alpha_{i+1}, \A_{i+1}$, that is, 
$\lang{A}\in\hlang{\A_i}$ iff $\lang{A}\in\hlang{\A_{i+1}}$. 
}

If $\quant_i = \forall$, then  $\lang{A}\in \lang{\A_i}$ iff 
$\lang{A}\notin \overline{\hlang{\A_i}}$. The 
quantification condition of $\overline{\A_i}$ begins with $\exists x_i$. We then construct $\A_{i+1}$ w.r.t. $\overline{\A_i}$ as described above, and check for non-membership.

Every $\forall$ quantifier requires complementation, which is exponential in 
$|Q|$. Therefore, in the worst case, the complexity of this algorithm is 
$O(2^{2^{...^{|Q||A|}}})$, where the tower is of height $k$. If the number of 
$\forall$ quantifiers is fixed, then the complexity is $O(|Q||A|^k)$. 
\end{proof}

Since nonemptiness of NFH is undecidable, so are its universality and containment problems. However, we show that containment is decidable for the fragments that we consider.

\begin{theorem}
\label{thrm:containment}
 The containment problems of $\nfhe$ and $\nfhf$ in $\nfhe$ and $\nfhf$ and of 
$\nfhef$ in $\nfhe$ and $\nfhf$ are \comp{PSPACE-complete}.
\end{theorem}

\begin{proof}[Proof Sketch]
The lower bound follows from the \comp{PSPACE-hardness} of the containment problem for NFA. 
For the upper bound, for two NFH $\A_1$ and $\A_2$, we have that 
$\hlang{\A_1}\subseteq\hlang{\A_2}$ iff 
$\hlang{\A_1}\cap\overline{\hlang{\A_2}}  =  \emptyset$. 
We can use the constructions in the proof of Theorem~\ref{thm:nfh.operations} 
to compute a matching NFH $\A = 
\A_1\cap\overline{\A_2}$, and check its nonemptiness. Complementing $\A_2$ is exponential in its number of states, and the intersection construction is polynomial. 

If $\A_1\in\nfhe$ and $\A_2\in\nfhf$ or vice versa, then $\A$ is an $\nfhe$ or 
$\nfhf$, respectively, whose nonemptiness can be decided in space that is 
logarithmic in $|\A|$.   

It follows from the construction in the proof of Theorem~\ref{thm:nfh.operations}, that the quantification condition of $\A$ may be any 
interleaving of the quantification conditions of the two intersected NFH.
Therefore, for the rest of the fragments, we can construct the intersection such that $\A$ is an $\nfhef$. 

The PSPACE upper bound of Theorem~\ref{thm:nfhef.nonemptiness} is derived from the number of variables and not from the state-space of the NFH. Therefore, while $|\bar{\A_2}|$ is exponential in the number of states of $\A_2$, checking the nonemptiness of $\A$ is in \comp{PSPACE}. 
\end{proof}

\section{Learning NFH}
\label{sec:learning}

In this section, we introduce $\lstar$-based learning algorithms for the fragments $\nfhf$ and 
$\nfhe$. 
We first survey the $\lstar$ algorithm \cite{Angluin87}, and then describe the relevant 
adjustments for our case.

\subsection{Angluin's $\lstar$ Algorithm}

$\lstar$ consists of two entities: a {\em learner}, who wishes to learn a DFA $A$ for an unknown (regular) language $\cal L$, and a {\em teacher}, who knows $\cal L$.
During the learning process, the learner asks the teacher two types of queries: {\em membership queries} (``is the word $w$ in $\cal L$?'') and {\em equivalence queries} (``is $A$ a DFA for $\cal L$?'').

The learner maintains $A$ in the form of an {\em observation table} $T$ of truth values, whose rows $D, D\cdot\Sigma$ and columns $E$ are sets of words over $\Sigma$, where $D$ is prefix-closed, and $E$ is suffix-closed. Initially, $D = E = \{\epsilon\}$. 
For a row $d$ and a column $e$, the entry for $T(d,e)$ is $\true$ iff $d\cdot e \in{\cal L}$. The entries are filled via membership queries.
The vector of truth values for row $d$ is denoted $\row(d)$. Intuitively, the rows in $D$ determine the states of $A$, and the rows in $D\cdot\Sigma$ determine the transitions of $A$: the state $\row(d\cdot\sigma)$ is reached from $\row(d)$ upon reading $\sigma$. 

The learner updates $T$ until it is {\em closed}, which, intuitively, ensures a full transition relation and {\em consistent}, which, intuitively, ensures a deterministic transition relation.
If $T$ is not closed or not consistent then more rows or more columns are added to $T$, respectively. 

\stam{
The learner updates the table until it is both {\em closed} and {\em consistent}. The table is closed if for every $d\in D$ and $\sigma\in \Sigma$, there is $d'\in D$ such that $\row(d') = \row(d\cdot\sigma)$. Intuitively, a closed table assures a full transition relation.
The table is consistent if for every $d_1,d_2\in D$ and every $\sigma \in \Sigma$, it holds that if $\row(d_1) = \row(d_2)$, then $\row(d_1\cdot\sigma) = \row(d_2\cdot\sigma)$. Intuitively, a consistent table assures a deterministic transition relation. 

If the table is not closed, then the missing row $d'$ is added to $D$, and $d'\cdot\sigma$ is added to $D\times\Sigma$ for every $\sigma\in \Sigma$. The new entries in $T$ are filled via membership queries. Note that this leaves $D$ prefix-closed.
If the table is inconsistent, then there is $e\in E$ for which $d_1\cdot\sigma\cdot e \neq d_2\cdot\sigma\cdot e$. The word $\sigma\cdot e$ then separates $\row(d_1)$ from $\row(d_2)$, and is added to $E$ (notice that this leaves $E$ suffix-closed). The new table entries are filled accordingly, and now $\row(d_1)\neq \row(d_2)$. 
}

When $T$ is closed and consistent, the learner constructs $A$: The states are the rows of $D$, the initial state is $\row(\epsilon)$, the accepting states are these in which $T(d,\epsilon) = \true$, and the transition relation is as described above. The learner then submits an equivalence query. If the teacher confirms, the algorithm terminates. Otherwise, the teacher returns a counterexample $w\in \lang{A}$ but $w\notin {\cal L}$ (which we call a {\em positive counterexample}), or $w\notin \lang{A}$ but $w\in {\cal L}$ (which we call a {\em negative counterexample}). The learner then adds $w$ and all its suffixes to $E$, and proceeds to construct the next candidate DFA $A$. 

It is shown in \cite{Angluin87} that as long as $A$ is not a DFA for $\cal L$, it has less states than a minimal DFA for $\cal L$. Further, every change in the table adds at least one state to $A$. Therefore, the procedure is guaranteed to terminate successfully with a minimal DFA $A$ for $\cal L$. 

The correctness of the $\lstar$ algorithm follows from the fact that regular languages have a {\em canonical form}, which guarantees a single minimal DFA for a regular language $\cal L$. To enable an $\lstar$-based algorithm for  $\nfhf$ and $\nfhe$, we first define canonical forms for these fragments.

\subsection{Canonical Forms for the Alternation-Free Fragments}\label{subsec:canonical.forms}

We begin with the basic terms on which our canonical forms are based.
\begin{definition}
\begin{enumerate}
\item An $\nfhf$ $\A_\forall$ is {\em sequence complete} if for every word ${\bi w}$, it holds that $\hat{\A_\forall}$ accepts ${\bi w}$ iff it accepts every sequence of ${\bi w}$. 
\item An $\nfhe$ $\A_\exists$ is {\em permutation complete} if for every word 
${\bi w}$, it holds that $\hat \A_\exists$ accepts ${\bi w}$ iff it accepts every permutation of 
${\bi w}$. 
\end{enumerate}
\end{definition}

An $\nfhf$ $\A_\forall$ accepts a hyperword $S$ iff $\hat\A_\forall$ accepts every sequence of size $k$ of $S$. If some sequence is missing from $\lang{\hat\A}$, then removing the rest of the sequences of $S$ from $\lang{\hat\A_\forall}$ does not affect the non-acceptance of $S$. Therefore, the underlying automata of sequence-complete $\nfhf$ only accept necessary sequences. 
Similarly, an $\nfhe$ $\A_\exists$ accepts a hyperword $S$ iff $\hat\A_\exists$ accepts some permutation $p$ of size $k$ of words in $S$. Adding the rest of the permutations of $p$ to $\lang{\hat\A_\exists}$ does not affect the acceptance of $S$. Therefore, the underlying automata of permutation-complete $\nfhe$ only reject the necessary permutations of every hyperword. 
As a conclusion, we have the following.

\begin{lemma}
\label{lem:langeq}
\begin{enumerate}
\item Let $\A_\forall$ be an $\nfhf$, and let $\A'_\forall$ be a sequence-complete $\nfhf$ over $\Sigma$ and $X$ such that for every word ${\bi w}$, the underlying NFA
$\hat{\A'_\forall}$ accepts ${\bi w}$ iff $\hat{\A_\forall}$ accepts every sequence of ${\bi w}$. Then $\hlang{\A_\forall} =\hlang{\A'_\forall}$.
\item Let $\A_\exists$ be an $\nfhe$, and let $\A'_\exists$ be a permutation-complete $\nfhe$ over $\Sigma$ and $X$ such that for every word ${\bi w}$, the underlying NFA $\hat{\A_\exists}$ accepts ${\bi w}$ iff $\hat{\A'_\exists}$ accepts all permutations of ${\bi w}$. 
Then $\hlang{\A_\exists} =\hlang{\A'_\exists}$.
\end{enumerate}

\end{lemma}

Next, we show that we can construct a sequence- or permutation-complete NFH for a given $\nfhf$ or $\nfhe$, respectively. Intuitively, given $\A$,
for every sequence (permutation) $\zeta$ of $(1,\ldots k)$, we construct an NFA that runs on ${\bi w}_\zeta$ in the same way that $\hat\A$ runs on $\bi w$, for every $\bi w$.
The underlying NFA we construct for the $\nfhf$ and $\nfhe$ are the intersection and union, respectively, of all these NFA. 

\begin{lemma}\label{lem:permutation.sequence.complete}
Every $\nfhf$ ($\nfhe$) $\A$ has an equivalent sequence-complete (permutation-complete) $\nfhf$ ($\nfhe$) $\A'$ over the same set of variables. 
\end{lemma}

Finally, as the following theorem shows, sequence- and permutation- complete NFH offer a unified model for the alternation-free fragments.

\begin{theorem}\label{thm:permutation.sequence.complete}
Let $\A_1$ and $\A_2$ be two sequence-complete (permutation-complete) $\nfhf$ ($\nfhe$) over the same set of variables. Then $\hlang{\A_1}=\hlang{\A_2}$ iff  $\lang{\hat\A_1} = \lang{\hat \A_2}$.
\end{theorem}

\stam{
\begin{lemma}
\label{lem:eqseq}
Every $\nfhf$ $\A$ has an equivalent sequence-complete $\nfhf$ $\A'$ over the 
same set of variables. 
\end{lemma}

\begin{lemma}
\label{lem:seqcomp}
Let $\A_1$ and $\A_2$ be two sequence-complete $\nfhf$ over the same set of 
variables $X$.
Then $\lang{\A_1}=\lang{\A_2}$ iff  $\lang{\hat\A_1} = \lang{\hat \A_2}$.
\end{lemma}
}

Regular languages have a canonical form, which are minimal DFA. We use this property to define canonical forms for $\nfhf$ and $\nfhe$ as sequence-complete (permutation-complete) $\nfhf$ ($\nfhe$) with a minimal number of variables and a minimal underlying DFA. 

\stam{

\subsubsection{A Canonical Form for $\nfhe$}

We say that an $\nfhe$ $\A'$ is {\em permutation complete} if for every word 
$w$, it holds that $\hat \A'$ accepts $w$ iff it accepts every permutation of 
$w$. 

\begin{lemma}
\label{lem:premcomp}
Let $\A$ be an $\nfhe$.
Let $\A'$ be a permutation-complete $\nfhe$ over $\Sigma$ and $X$ with the following property: for every word $w$, the underlying NFA $\hat\A$ accepts a word $w$ iff $\hat\A'$ accepts all permutations of $w$. 
Then $\lang{\A} =\lang{\A'}$.
\end{lemma}

\begin{lemma}
\label{lem:eqpermcomp}
Every $\nfhe$ has an equivalent permutation-complete $\nfhe$ over the same set 
of variables. 
\end{lemma}

\begin{lemma}
\label{lem:permcompsame}
Let $\A_1$ and $\A_2$ be two permutation-complete $\nfhe$ over the same set of 
variables $X$.
Then $\lang{\A_1}=\lang{\A_2}$ iff  $\lang{\hat\A_1} = \lang{\hat \A_2}$.
\end{lemma}

We define a canonical form for $\nfhe$ as a minimal deterministic 
permutation-complete $\nfhe$ with a minimal number of variables.

}

\subsection{Learning $\nfhf$ and $\nfhe$}

We now describe our $\lstar$-based learning algorithms for $\nfhe$ and $\nfhf$.
These algorithms aim to learn an NFH with the canonical form defined in Section~\ref{subsec:canonical.forms} for a target hyperlanguage $\hl$. Figure~\ref{fig:learning_flow} presents the overall flow of the learning algorithms for both fragments.

In the case of hyperautomata, the membership queries and the counterexamples provided by the teacher consist of hyperwords. Similarly to \cite{fht19}, we assume a teacher that returns a minimal counterexample in terms of size of the hyperword. 

During the procedure, the learner maintains an NFH $\A$ via an observation table for $\hat\A$, over the alphabet $\hat\Sigma  = (\Sigma\cup\{\#\})^k$, where $k$ is initially set to $1$. 
When the number of variables is increased to $k'>k$, the alphabet of $\hat\A$ is extended accordingly to $(\Sigma\cup\{\#\})^{k'}$. 
To this end, we define a function $\uparrow_k^{k'}:(\Sigma\cup\{\#\})^k\rightarrow (\Sigma\cup\{\#\})^{k'}$, which replaces every letter $(\sigma_1,\ldots \sigma_k)$, with $(\sigma_1,\ldots \sigma_k)+(\sigma_k)^{k'-k}$. That is, the last letter is duplicated to create a $k'$-tuple. 
We extend $\uparrow_k^{k'}$ to words: $\uparrow_k^{k'}({\bi w})$ is obtained by replacing every letter $\sigma$ in ${\bi w}$ with $\uparrow_k^{k'}(\sigma)$. 
Notice that, for both fragments, if $\unzip(d\cdot e) \in \hlang{\A}$, then $\unzip(\uparrow_k^{k'}(d\cdot e)) \in \hlang{\A}$. 
Accordingly, when the number of variables is increased, every word ${\bi w}$ in the rows and columns of $T$ is replaced with $\uparrow_k^{k'}({\bi w})$, an action which we denote by $\uparrow_k^{k'}(T)$.

\subsubsection{Learning $\nfhf$}

In the case of $\nfhf$, when the teacher returns a counterexample $S$, it holds 
that if $|S|> k$, then $S$ must be positive. Indeed, assume by way of 
contradiction that $S$ is negative. Then, for every $k$ words $w_1,\ldots, w_k$ 
in $S$, it holds that $\zip(w_1,\ldots, w_k)\in \lang{\hat\A}$, but $S\notin 
\hl$. Therefore, in an $\nfhf$ $\A'$ for $\hl$, there exists some word of the 
form $w = \zip(w_1,\ldots w_k)$ such that $w_i\in S$ for $1\leq i \leq k$, and 
$w\notin \lang{\hat\A'}$. As a result, $\{w_1,\ldots, w_k\}\notin \hl$. Since 
$\zip(w_1,\ldots, w_k)$ and all its sequences are in $\lang{\hat\A}$, then a 
smaller counterexample is $\{w_1,\ldots, w_k\}$, a contradiction to the 
minimality of $S$. 

In fact, if $|S|>k$, then it must be that $|S| = k+1$. Indeed, since $S$ is a 
positive counterexample, and $\A$ accepts all representations of subsets of size 
$k$ of $S$ (otherwise the teacher would return a counterexample of size $k$), 
then there exists a subset $S'\subseteq S$ of size $k+1$ that should be 
represented, but is not. Therefore, $S'$ is a counterexample of size $k+1$.

When a counterexample $S$ of size $k+1$ is returned, the learner updates $k\leftarrow k+1$, updates $T$ to $\uparrow_k^{k+1}(T)$, arbitrarily selects a permutation $p$ of the words in $S$, and adds  $\zip(p)$ and all its suffixes to $E$.
In addition, it updates $D\cdot\hat\Sigma$ in accordance with the new updated $\hat\Sigma$, and fills in the missing entries.  

When $|S| \leq k$, then the counterexample is either positive or negative.
If $S$ is positive, then there exists some permutation $p$ of the words in $S$ such that $\A$ does not accept $\zip(p)$ (a permutation and not a proper sequence, or there would be a smaller counterexample). The learner finds such a permutation $p$, and adds $\zip(p)$ and all its suffixes to $E$. Notice that $\zip(p)$ does not already appear in $T$, since a membership query would have returned ``yes'', and so $\hat\A$ would have accepted $\zip(p)$.

if $S$ is negative, then $\A$ accepts all sequences of length $k$ of words in $S$, though it should not.
Then there exists a permutation $p$ of the words in $S$ that does not appear in $T$, and which $\A$ accepts. The learner then finds such a permutation $p$ and adds $\zip(p)$ and all its suffixes to $E$.

If $p$ is a permutation of the words in $S$, and $S$ is a negative counterexample, then $\zip(p)$ should not be in $\lang{\hat\A}$ due to any other hyperword, and if $S$ is a positive counterexample, then it should be in $\lang{\hat\A}$ for every $S'$ such that $S\subseteq S'$. Therefore, the above actions by the learner are valid.  

When an equivalence query succeeds, then $\A$ is indeed an $\nfhf$ for $\hl$. 
However, $\A$ is not necessarily sequence-complete, as $\hat\A$ may accept a 
word ${\bi w} = \zip(w_1,\ldots, w_k)$ but not all of its sequences. This check 
can be performed by the learner directly on $\hat\A$. 
Notice that ${\bi w}$ does not occur in $T$, since a membership query on ${\bi w}$ would return ``no''. 
Once it is verified that $\A$ is not sequence-complete, the counterexample ${\bi w}$ (and all its suffixes) are added to $E$, and the procedure returns to the learning loop.
 
As we have explained above, variables are added only when necessary, and so the output $\A$ is indeed an NFH for $\hl$ with minimally many variables. 
The correctness of $\lstar$ and the minimality of the counterexamples returned by the teacher guarantee that for each $k'\leq k$, the run learns a minimal deterministic $\hat\A$ for hyperwords in $\hl$ that are represented by $k'$ variables. Therefore, a smaller $\hat\A'$ for $\hl$ does not exist, as restricting $\hat\A'$ to the first $k'$ letters in each $k$-tuple would produce a smaller underlying automaton for $k'$ variables, a contradiction. 

\begin{figure}[ht]
    \begin{center}
        \includegraphics[scale=0.5]{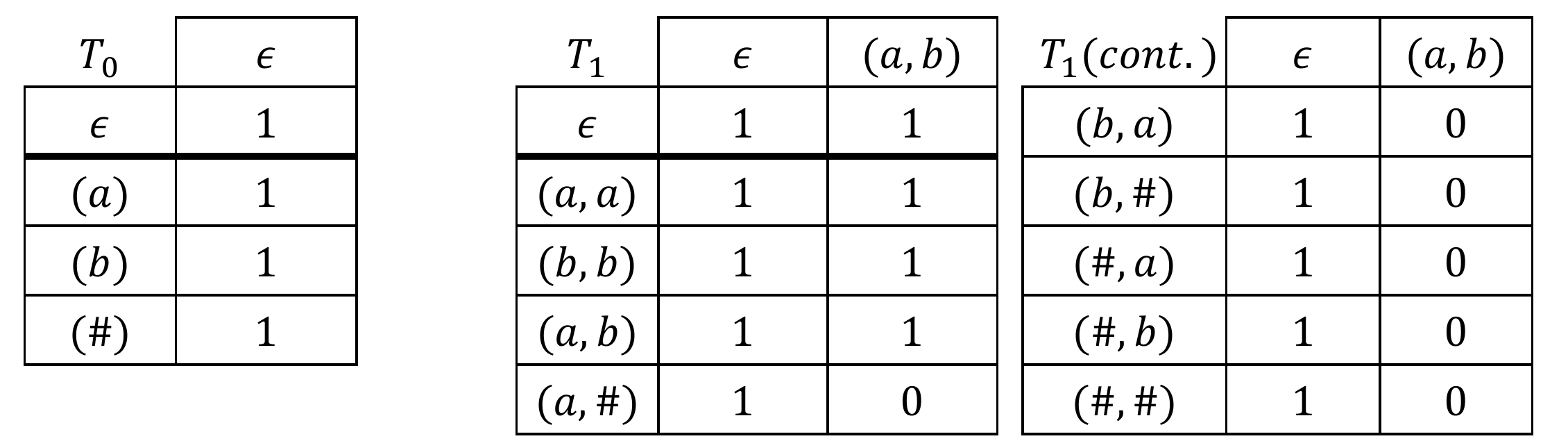}
    \end{center}
    \caption{The first stages of learning $\hlang{\A_3}$ of Figure~\ref{fig:nfh_examples}.}
    \label{fig:learning_nfhf}
\end{figure}

\begin{example}
Figure~\ref{fig:learning_nfhf} displays the first two stages of learning $\hlang{\A_3}$ of Figure~\ref{fig:ordered}.
$T_0$ displays the initial table, with $D=E = \{\epsilon\}$, and $\hat\Sigma = \{a,b,\#\}$. since $\{a\}, \{b\}$, and $\{\epsilon\}$ are all in $\hlang{\A_3}$, the initial candidate NFH $\A$ includes a single variable, and, following the answers to the membership queries, a single accepting state.

Since $\hlang{\A_3}$ includes all hyperwords of size $1$, which are now accepted by $\A$, the smallest counterexample the teacher returns is of size $2$, which, in the example, is $\{a,b\}$. Table $T_1$ is then obtained from $T_0$ by applying $\uparrow_1^2$, updating the alphabet $\hat\Sigma$ to $\{a,b,\#\}^2$, and updating $D\cdot\hat\Sigma$ accordingly. $T_1$ is filled by submitting membership queries. For example, for $(b,a)\in D\cdot\hat\Sigma$ and $(a,b)\in E$, the learner submits a membership query for $\{ba, ab\}$, to which the teacher answers ``no''.
\end{example}

\subsubsection{Learning $\nfhe$}

The learning process for $\nfhe$ is similar to the one for $\nfhf$. We briefly describe the differences.

As in $\nfhf$, relying on the minimality of the counterexamples returned by the teacher guarantees that when a counterexample $S$ such that $|S|>k$ is returned, it is a positive counterexample. 
Indeed, assume by way of contradiction that $S$ is a negative counterexample of 
size $k'$. Since $\hat\A$ accepts $S$, there exists a word $\zip(w_1,\ldots, 
w_k)$ in $\lang{\hat\A}$ such that $\{w_1,\ldots, w_k\}\in S$. According to the 
semantics of $\exists$, if $\zip(w_1,w_2,\ldots, w_k)\in\lang{\hat\A}$ then 
$S\in\hlang\A$. Since $S\notin\hl$, we have that $\{w_1,\ldots, w_k\}$ is a 
smaller counterexample, a contradiction. 

Therefore, when the teacher returns a counterexample $S$ of size $k'>k$, the alphabet $\hat\Sigma$ is extended to $(\Sigma\cup\{\#\})^{k'}$, and the table $T$ is updated by $\uparrow_{k}^{k'}$, as is done for $\nfhf$.

If $|S|\leq k$, then $S$ may be either positive or negative. If $S$ is negative, then there exists some permutation of $S$ that is accepted by $\hat\A$. However, no such permutation is in $T$, as a membership query would have returned ``no''. Similarly, if $S$ is positive, then there exists no permutation of $S$ that $\hat\A$ accepts. In both cases, the learner chooses a permutation of $S$ and adds it, and all its suffixes, to $E$. 

As in the case of $\nfhf$, the success of an equivalence query does not necessarily imply that $\A$ is permutation-complete. 
If $\A$ is not permutation-complete, the learner finds a word ${\bi w}$ that is a permutation of ${\bi w}'$ such that ${\bi w}'\in\lang{\hat\A}$ but ${\bi w}\notin\lang{\hat\A}$, and adds ${\bi w}$ as a counterexample to $E$. 
The procedure then returns to the learning loop.

\begin{figure} 

\centering

\tikzset{every picture/.style={line width=0.75pt}} 

\scalebox{.95}{
\begin{tikzpicture}[x=0.75pt,y=0.75pt,yscale=-1,xscale=1]

\draw    (78.05,195.76) .. controls (77.68,166.57) and (91.81,174.17) .. (90.87,193.02) ;
\draw [shift={(90.63,195.76)}, rotate = 276.98] [fill={rgb, 255:red, 0; green, 0; blue, 0 }  ][line width=0.08]  [draw opacity=0] (8.93,-4.29) -- (0,0) -- (8.93,4.29) -- cycle    ;

\draw   (285.46,140.05) .. controls (285.46,135.9) and (288.82,132.53) .. (292.98,132.53) -- (422.81,132.53) .. controls (426.96,132.53) and (430.33,135.9) .. (430.33,140.05) -- (430.33,162.59) .. controls (430.33,166.74) and (426.96,170.11) .. (422.81,170.11) -- (292.98,170.11) .. controls (288.82,170.11) and (285.46,166.74) .. (285.46,162.59) -- cycle ;
\draw    (360,195) -- (360,174) ;
\draw [shift={(360,171)}, rotate = 450] [fill={rgb, 255:red, 0; green, 0; blue, 0 }  ][line width=0.08]  [draw opacity=0] (8.93,-4.29) -- (0,0) -- (8.93,4.29) -- cycle    ;

\draw    (360.15,232.33) -- (360.15,254.42) ;
\draw [shift={(360.15,257.42)}, rotate = 270] [fill={rgb, 255:red, 0; green, 0; blue, 0 }  ][line width=0.08]  [draw opacity=0] (8.93,-4.29) -- (0,0) -- (8.93,4.29) -- cycle    ;

\draw    (318.59,285.97) -- (147.88,285.97) ;
\draw [shift={(144.88,285.97)}, rotate = 360] [fill={rgb, 255:red, 0; green, 0; blue, 0 }  ][line width=0.08]  [draw opacity=0] (8.93,-4.29) -- (0,0) -- (8.93,4.29) -- cycle    ;

\draw    (88.79,255.9) -- (88.97,236) ;
\draw [shift={(89,233)}, rotate = 450.53] [fill={rgb, 255:red, 0; green, 0; blue, 0 }  ][line width=0.08]  [draw opacity=0] (8.93,-4.29) -- (0,0) -- (8.93,4.29) -- cycle    ;

\draw    (140,216) -- (183.56,216) ;
\draw [shift={(186.56,216)}, rotate = 180] [fill={rgb, 255:red, 0; green, 0; blue, 0 }  ][line width=0.08]  [draw opacity=0] (8.93,-4.29) -- (0,0) -- (8.93,4.29) -- cycle    ;

\draw    (280.98,215.74) -- (316,215.98) ;
\draw [shift={(319,216)}, rotate = 180.4] [fill={rgb, 255:red, 0; green, 0; blue, 0 }  ][line width=0.08]  [draw opacity=0] (8.93,-4.29) -- (0,0) -- (8.93,4.29) -- cycle    ;

\draw    (318.59,278.87) -- (141.02,231.11) ;
\draw [shift={(138.12,230.33)}, rotate = 375.05] [fill={rgb, 255:red, 0; green, 0; blue, 0 }  ][line width=0.08]  [draw opacity=0] (8.93,-4.29) -- (0,0) -- (8.93,4.29) -- cycle    ;

\draw   (318.59,203.21) .. controls (318.59,199.09) and (321.93,195.76) .. (326.04,195.76) -- (404.26,195.76) .. controls (408.37,195.76) and (411.71,199.09) .. (411.71,203.21) -- (411.71,225.56) .. controls (411.71,229.67) and (408.37,233) .. (404.26,233) -- (326.04,233) .. controls (321.93,233) and (318.59,229.67) .. (318.59,225.56) -- cycle ;
\draw    (285.46,140.05) -- (142.89,199.6) ;
\draw [shift={(140.12,200.76)}, rotate = 337.33000000000004] [fill={rgb, 255:red, 0; green, 0; blue, 0 }  ][line width=0.08]  [draw opacity=0] (8.93,-4.29) -- (0,0) -- (8.93,4.29) -- cycle    ;

\draw    (431,147.86) -- (487,147.86) ;
\draw [shift={(490,147.86)}, rotate = 180] [fill={rgb, 255:red, 0; green, 0; blue, 0 }  ][line width=0.08]  [draw opacity=0] (8.93,-4.29) -- (0,0) -- (8.93,4.29) -- cycle    ;

\draw   (186.68,203.21) .. controls (186.68,199.09) and (190.02,195.76) .. (194.13,195.76) -- (272.35,195.76) .. controls (276.46,195.76) and (279.8,199.09) .. (279.8,203.21) -- (279.8,225.56) .. controls (279.8,229.67) and (276.46,233) .. (272.35,233) -- (194.13,233) .. controls (190.02,233) and (186.68,229.67) .. (186.68,225.56) -- cycle ;
\draw   (318.59,263.62) .. controls (318.59,259.51) and (321.93,256.17) .. (326.04,256.17) -- (404.26,256.17) .. controls (408.37,256.17) and (411.71,259.51) .. (411.71,263.62) -- (411.71,285.97) .. controls (411.71,290.08) and (408.37,293.42) .. (404.26,293.42) -- (326.04,293.42) .. controls (321.93,293.42) and (318.59,290.08) .. (318.59,285.97) -- cycle ;
\draw   (47.01,203.54) .. controls (47.01,199.42) and (50.35,196.09) .. (54.46,196.09) -- (132.68,196.09) .. controls (136.79,196.09) and (140.12,199.42) .. (140.12,203.54) -- (140.12,225.89) .. controls (140.12,230) and (136.79,233.33) .. (132.68,233.33) -- (54.46,233.33) .. controls (50.35,233.33) and (47.01,230) .. (47.01,225.89) -- cycle ;
\draw   (50,263.2) .. controls (50,259.09) and (53.34,255.76) .. (57.45,255.76) -- (135.67,255.76) .. controls (139.78,255.76) and (143.11,259.09) .. (143.11,263.2) -- (143.11,285.55) .. controls (143.11,289.66) and (139.78,293) .. (135.67,293) -- (57.45,293) .. controls (53.34,293) and (50,289.66) .. (50,285.55) -- cycle ;
\draw    (29.79,211.24) -- (43.84,211.24) ;
\draw [shift={(46.84,211.24)}, rotate = 180] [fill={rgb, 255:red, 0; green, 0; blue, 0 }  ][line width=0.08]  [draw opacity=0] (8.93,-4.29) -- (0,0) -- (8.93,4.29) -- cycle    ;

\draw (141.81,182.12) node   [align=left] {$ $};
\draw (93.57,214.71) node   [align=left] {$\displaystyle T${\fontfamily{ptm}\selectfont  }~closed \\and consistent?};
\draw (98.22,162.32) node   [align=left] {no: membership queries};
\draw (233.24,214.38) node   [align=left] {construct $\displaystyle \mathcal{A}$};
\draw (357.89,151.32) node   [align=left] {permutation/sequence \\complete?};
\draw (229.13,150.68) node  [rotate=-337.24] [align=left] {no: add cex~{\fontfamily{ptm}\selectfont  }$\displaystyle w$};
\draw (365.15,214.38) node   [align=left] {equivalent?};
\draw (394.44,242.91) node  [rotate=-359.45] [align=left] {no: cex $\displaystyle S$};
\draw (365.15,274.8) node  [rotate=-359.45] [align=left] {$\displaystyle |S| >k?$};
\draw (378.6,183.41) node  [rotate=-359.45] [align=left] {yes};
\draw (163.6,207.37) node  [rotate=-359.45] [align=left] {yes};
\draw (291.5,204) node   [align=left] {{\fontfamily{ptm}\selectfont  }$\displaystyle \mathcal{A}$};
\draw (269.92,254.66) node  [rotate=-14.77] [align=left] {no: add cex $\displaystyle S${\fontfamily{ptm}\selectfont  }};
\draw (237.84,277.89) node  [rotate=-359.61] [align=left] {yes};
\draw (458.14,137.71) node  [rotate=-359.45] [align=left] {yes: {\fontfamily{ptm}\selectfont  }$\displaystyle \mathcal{A}$};
\draw (115.57,245.15) node  [rotate=-359.61] [align=left] {add~{\fontfamily{ptm}\selectfont  }$\displaystyle S${\fontfamily{ptm}\selectfont  }};
\draw (96.56,274.38) node   [align=left] {{\fontfamily{ptm}\selectfont  }$\displaystyle T\leftarrow \uparrow ^{k'}_{k}( T)${\fontfamily{ptm}\selectfont  }};

\end{tikzpicture}
}

\caption{The learning process flow for $\nfhf$ and $\nfhe$.}
\label{fig:learning_flow}
\vspace{-3mm}
\end{figure}

\section{Conclusion and Future Work}
\label{sec:concl}

We have introduced and studied {\em hyperautomata} and {\em hyperlanguages}, 
focusing on 
the basic model of regular hyperlanguages, in which the underlying automaton is 
a standard NFA. We have shown that regular hyperlanguages are closed under 
set operations (complementation, intersection, and union) and are capable of 
expressing important hyperproperties for information-flow security policies over 
finite traces. We have also investigated fundamental decision procedures such as 
checking nonemptiness and membership. We have shown that their regular 
properties allow the learnability of the alternation-free fragments. Fragments 
that combine the two types of quantifiers prove to be more challenging, and we 
leave their learnability to future work.

The notion of hyperlanguages, as well as the model of hyperautomata, can be 
lifted to handle hyperwords that consist of {\em infinite} words: instead 
of an underlying finite automaton, we can use any model that accepts infinite 
words. In fact, we believe using an underlying alternating B{\"u}chi automaton, 
such hyperautomata can express the entire logic of HyperLTL~\cite{cfkmrs14}, 
using the standard Vardi-Wolper construction for LTL~\cite{VW94} as 
basis. Our complexity results for the various decision procedures for NFH, 
combined with the complexity results shown in~\cite{fh16}, suggest that using 
hyperautomata would be optimal, complexity-wise, for handling HyperLTL.

Further future directions include studying non-regular hyperlanguages 
(e.g., context-free), and object hyperlanguages (e.g., trees). Other open 
problems include a full investigation of the complexity of decision procedures 
for alternating fragments of NFH.

\bibliography{bibliography,extra_bib}

\newpage
\appendix

\begin{center}

{\huge \bf Appendix}
 
\end{center}

\section{Proofs}

\subsection*{Theorem~\ref{thm:nfh.operations}}

\begin{proof}
{\bf Complementation. }
Let $\A$ be an NFH. The NFA $\hat \A$ 
can be complemented with respect to its language over $\hat\Sigma$ to an NFA 
$\overline{\hat{\A}}$. 
Then for every assignment $v:X\rightarrow S$, it holds that $\hat \A$ accepts 
$\zip(v)$ iff $\overline{\hat{\A}}$ does not accept $\zip(v)$.
Let $\overline{\alpha}$ be the quantification condition obtained from $\alpha$ 
by 
replacing every $\exists$ with $\forall$ and vice versa. 
We can prove by induction on $\alpha$ that $\overline{\A}$, the NFH 
whose 
underlying NFA is $\overline{\hat{\A}}$, and whose quantification 
condition is 
$\overline{\alpha}$, accepts $\overline {\hlang{\A}}$.
The size of $\overline{\A}$ is exponential in $|Q|$, due to the complementation construction for $\hat\A$.

Now, let $\A_1 = 
\tuple{\Sigma,X,Q,Q_0,\delta_1,F_1,\alpha_1}$ and $\A_2= 
\tuple{\Sigma,Y,P,P_0,\delta_2,F_2,\alpha_2}$ be two NFH with $|X|=k$ and 
$|Y|=k'$ variables, respectively.

\bigbreak
\noindent{\bf Union.}
We construct an NFH 
$\A_{\cup} = \tuple{\Sigma,X\cup Y, Q\cup P\cup\{p_1,p_2\},  Q_0\cup P_0, \delta, 
F_1\cup F_2\cup\{p_1,p_2\}, \alpha}$, where $\alpha = \alpha_1\alpha_2$ (that is, we concatenate the two quantification conditions), and where $\delta$ is defined as follows. 

\begin{itemize}
\item For every $$(q_1\xrightarrow {(\sigma_1,\ldots, \sigma_k)} q_2)\in 
\delta_1$$ we set $$(q_1\xrightarrow{(\sigma_1,\ldots, \sigma_k)+t} q_2)\in 
\delta$$ for every $t\in (\Sigma\cup\{\#\})^{k'}$. 

\item
For every $$(q_1\xrightarrow {(\sigma_1,\ldots, \sigma_{k'})} q_2)\in 
\delta_2$$ we set $$(q_1\xrightarrow{t+(\sigma_1,\ldots, \sigma_{k'})} q_2)\in 
\delta$$ for every $t\in (\Sigma\cup\{\#\})^{k}$. 

\item
For every $q\in F_1$, we set 
$$(q\xrightarrow{(\#)^k+t}p_1)\in \delta$$
and
$$(p_1\xrightarrow{(\#)^k+t}p_1)\in \delta$$
for every $t\in (\Sigma\cup\{\#\})^{k'}$.

\item
For every $q\in F_2$, we set 
$$(q\xrightarrow{t+(\#)^{k'}}p_2)\in \delta$$
and
$$(p_2\xrightarrow{t+(\#)^{k'}}p_2)\in \delta$$
for every $t\in (\Sigma\cup\{\#\})^{k}$.
\end{itemize}

Let $S$ be a hyperword. 
For every $v:(X\cup Y)\rightarrow S$, it holds that if $\zip(v|_X)\in\lang{\hat\A_1}$, then $\zip(v)\in\lang{\hat\A_{\cup}}$. Indeed, according to our construction, every word assigned to the $Y$ variables is accepted in the $\A_1$ component of the construction, and so it satisfies both types of quantifiers.
A similar argument holds for $v|_Y$ and $\A_2$.

Also, according to our construction, for every $v:(X\cup Y)\rightarrow S$,  if $\zip(v)\in\lang{\hat\A_{\cup}}$, then either $\zip(v|_X)\in\lang{\hat\A_1}$, or $\zip(v|_Y)\in\lang{\hat\A_2}$. 
As a conclusion, we have that $\hlang{\A_{\cup}} = \hlang{\A_1}\cup\hlang{\A_2}$. 

The state space of $\A_\cup$ is linear in the state spaces of $\A_1,\A_2$. 
However, the size of the alphabet of $\A_\cup$ may be exponentially larger than 
that of $\A_1$ and $\A_2$, since we augment each letter with all sequences of 
size $k'$ (in $\A_1$) and $k$ (in $\A_2$).

\stam{
The proof is similar to the standard constructive proof for NFAs. Let $\A = 
\tuple{\Sigma,X,Q,Q_0,\delta,F,\alpha}$ with underlying NFA and $\A'= 
\tuple{\Sigma',X',Q',Q'_0,\delta',F',\alpha'}$ be two NHFs. We now construct an 
NHF $\A'' = \tuple{\Sigma'',X'',Q'',Q''_0,\delta'',F'',\alpha''}$ that accepts 
$\hlang{\A} \cup \hlang{\A'}$ as follows. Let
\begin{itemize}
 \item  $\Sigma'' = \Sigma \cup \Sigma' \cup \{\#\}$,
 \item $X'' = X \cup X'$,
 \item $Q'' = Q \cup Q' \cup \{q_{\mathit{init}}\}$, \todo{Sarai: we allow a set 
of inital state, so taking the union is enough, no need for a new state}
 \item $Q''_0 = q_{\mathit{init}}$,
 \item $F'' = F \cup F'$,
 \item $\alpha'' = \alpha.\alpha'$
 
\end{itemize}
Transitions of $\A''$ include the following:

\begin{itemize}
 \item  We add transitions $q\xrightarrow{\epsilon} q'$, for each $q' \in 
Q''_0$, where $\epsilon$ is the null letter. 

\item For each transition $q \xrightarrow{(\sigma_{i_1}, \dots, \sigma_{i_k})} 
q' \in \delta$ in the underlying NFA of $\A$, we include a transition $q 
\xrightarrow{(\sigma_{i_1}, \dots, \sigma_{i_k}, \#, \dots, \#)} q'$ in 
$\delta''$, where $\#$ is repeated $|\alpha'|$ number of times (i.e., the 
number 
of quantifiers in $\alpha$). Likewise, for each transition $q 
\xrightarrow{(\sigma_{i_1}, \dots, \sigma_{i_{k'}})} q' \in \delta'$ (in the 
underlying NFA of $\A'$), we include a transition $q \xrightarrow{(\#, \dots, 
\#, \sigma_{i_1}, \dots, \sigma_{i_{k'}})} q'$ in $\delta''$, where $\#$ is 
repeated $|\alpha|$ number of times (the number of quantifiers in $\alpha$).
\todo{Sarai: I think this is incorrect. We should not add $\#$, but every 
possible combination of $\Sigma\cup\{\#\}$ letters}
\end{itemize}

Now, it is straightforward to see that every hyperword in $\hlang{\A}$ or 
$\hlang{\A'}$ is accepted by $\A''$. This is because the quantifiers not 
relevant to hyperwords (e.g., $\alpha$) for a hyperword in $\hlang{\A'}$ are 
taken care of by the padded words in $(\#, \dots, \#, \sigma_{i_1}, \dots, 
\sigma_{i_{k'}})$. The same argument holds for quantifiers in $\alpha'$ and 
hyperwords in $\hlang{\A}$.
}

\bigbreak
\noindent{\bf Intersection.}
The proof follows the closure under union and complementation. However, we
also offer a direct translation, which avoids the need to complement. 
We construct an NFH 
$\A_\cap = \tuple{\Sigma,X\cup Y, (Q\cup \{q\} \times P\cup \{p\}), (Q_0\times P_0), \delta, 
(F_1\cup\{q\}) \times (F_2\cup\{p\}), \alpha_1\alpha_2}$, where $\delta$ is defined as follows. 
\begin{itemize}
    \item For every $(q_1\xrightarrow {(\sigma_1,\ldots, \sigma_k)} q_2)\in 
\delta_1$ and every $(p_1\xrightarrow {(\sigma'_1,\ldots, \sigma'_{k'}) 
}p_2)\in\delta_2$, we have 
$$\Big((q_1,p_1)\xrightarrow{(\sigma_1\ldots, \sigma_k,\sigma'_1,\ldots 
\sigma'_{k'})} (q_2,p_2)\Big)\in\delta$$
\item For every $q_1\in F_1, (p_1 \xrightarrow{(\sigma'_1,\ldots, 
\sigma'_{k'})} p_2)\in \delta_2$ we have
$$\Big((q_1,p_1)\xrightarrow {(\#)^k +(\sigma'_1,\ldots 
\sigma'_{k'})}(q,p_2)\Big), \Big((q,p_1)\xrightarrow {(\#)^k+(\sigma'_1,\ldots, 
\sigma'_{k'})}(q,p_2)\Big) \in \delta$$
\item For every $(q_1\xrightarrow {(\sigma_1,\ldots, \sigma_k)} q_2)\in 
\delta_1$ and $p_1\in F_2$, we have
$$\Big((q_1,p_1)\xrightarrow 
{(\sigma_1,\ldots, \sigma_{k})+(\#)^{k'}}(q_2,p)\Big), \Big((q_1,p)\xrightarrow 
{(\sigma_1,\ldots, \sigma'_{k})+(\#)^{k'}}(q_2,p)\Big)\in \delta$$
\end{itemize}
Intuitively, the role of $q,p$ is to keep reading $(\#)^k$ and $(\#)^{k'}$ after 
the word read by $\hat\A_1$ or $\hat\A_2$, respectively, has ended. 

The NFH $\hat{\A_\cap}$ simultaneously reads two words $\zip(w_1, w_2, \ldots, 
w_k)$ and $\zip(w'_1,w'_2,\ldots w'_{k'})$ that are read along $\hat\A_1$ and 
$\hat\A_2$, 
respectively, and accepts iff both words are accepted. The correctness 
follows from the fact that for $v:(X\cup Y)\rightarrow S$, we have that 
$\zip(v)$ is accepted by $\hat\A$ iff $\zip (v|_X)$ and $\zip(v|_Y)$ are 
accepted by $\hat\A_1$ and $\hat\A_2$, respectively. 

This construction is polynomial in the sizes of $\A_1$ and $\A_2$.
\end{proof}

\subsection*{Theorem~\ref{thm:nfh.nonemptiness}}

\begin{proof}
We mimic the proof idea in~\cite{fh16}, which uses a reduction from the {\em 
Post correspondence problem (PCP)}.
A PCP instance is a collection $C$ of dominoes of the form:
$$ \Bigg\{\Big[\frac{u_1}{v_1} \Big], 
\Big[\frac{u_2}{v_2} \Big],\dots, \Big[\frac{u_k}{v_k} \Big] \Bigg\}$$
where for all $i \in [1, k]$, we have $v_i,u_i \in 
\{a,b\}^*$.
%
The problem is to decide whether there exists a finite sequence of the dominoes 
of the form
$$\Big[\frac{u_{i_1}}{v_{i_1}} \Big]\Big[\frac{u_{i_2}}{v_{i_2}} \Big] \cdots 
\Big[\frac{u_{i_m}}{v_{i_m}} \Big]$$
where each index $i_j \in [1, k]$, such that the upper and lower finite 
strings of the dominoes are equal, i.e.,
$$ u_{i_1}u_{i_2}\cdots{}u_{i_m} = v_{i_1}v_{i_2}\cdots{}v_{i_m}$$
For example, if the set of dominoes is
$$ C_{\mathsf{exmp}} = \Bigg\{ \Big[\frac{ab}{b}\Big], 
\Big[\frac{ba}{a}\Big],\Big[\frac{a}{aba}\Big] \Bigg\} $$
Then, a possible solution is the following sequence of dominoes from 
$C_{\mathsf{exmp}}$: 
$$\mathsf{sol} = \Big[\frac{a}{aba}\Big]\Big[\frac{ba}{a}\Big] 
\Big[\frac{ab}{b}\Big ].  $$

\stam{
For some collection of dominoes, it is impossible to find a match. 
An example is 
the following:
$$\Big[\frac{w}{v} \Big] = \bigg\{\Big[\frac{abc}{ab} \Big], \Big[\frac{ca}{a} 
\Big], \Big[\frac{acc}{ba} \Big]\bigg\}$$

First, we map an arbitrary instance of the PCP 
$[\frac{w}{v}]=\{[\frac{w_2}{v_2}], [\frac{w_2}{v_2}], \dots, 
[\frac{w_k}{v_k}]\}$ to an NHF.
}

Given an instance $C$ of PCP, we encode a solution as a word $w_{sol}$ over the 
following alphabet:
$$\alphabet = \Big\{\frac{\sigma}{\sigma'} \mid \sigma,\sigma'\in\{a,b,{\dot 
a},{\dot b}, \$\}\Big\}.$$
Intuitively, $\dot{\sigma}$ marks the beginning of a new domino, and 
$\$$ marks the end of a sequence of the upper or lower parts of the dominoes 
sequence.


We note that $w_{sol}$ encodes a legal solution iff the following conditions are 
met:

\begin{enumerate}
    \item For every $ \frac{\sigma}{\sigma'}  $ that occurs in $w_{sol}$, it 
holds that $\sigma,\sigma'$ represent the same domino letter (both $a$ or both 
$b$, either dotted or undotted).
    \item The number of dotted letters in the upper part of $w_{sol}$ is equal 
to the number of dotted letters in the lower part of $w_{sol}$.
    \item $w_{sol}$ starts with two dotted letters, and the word $u_i$ between 
the $i$'th and $i+1$'th dotted letters in the upper part of $w_{sol}$, and the 
word $v_i$ between the corresponding dotted letters in the lower part of 
$w_{sol}$ are such that $ [\frac{u_i}{v_i}]  \in C$, for every $i$.
\end{enumerate}

We call a word that represents the removal of the first $k$ dominoes from 
$w_{sol}$ 
a {\em partial solution}, denoted by $w_{sol,k}$.
Note that the upper and lower parts of $w_{sol,k}$ are not necessarily of equal 
lengths (in terms of $a$ and $b$ sequences), since the upper and lower parts of 
a domino may be of different lengths, and so we use letter $\$$ to pad 
the end of the encoding in the shorter of the two parts. 

We construct an NFH $\A$, which, intuitively, expresses the following ideas: 
$(1)$ There exists an encoding $w_{sol}$ of a solution to $C$, and $(2)$ For 
every $w_{sol,k}\neq \epsilon$ in a hyperword $S$ accepted by $\A$, the word 
$w_{sol,k+1}$ is also in $S$. 

$\hlang{\A}$ is then the set of all hyperwords that contain an encoded solution 
$w_{sol}$, as well as all its suffixes obtained by removing a prefix of dominoes 
from $w_{sol}$. This ensures that $w_{sol}$ indeed encodes a legal solution. For example, a matching hyperword $S$ (for solution 
$\mathsf{sol}$ discussed earlier) that is accepted by $\A$ is:
$$ S = \{  w_{sol} = \frac{\dot a}{\dot a}   \frac{\dot b}{b}   \frac{a}{a}   
\frac{\dot a}{\dot a}   \frac{b}{\dot b} ,
w_{sol,1} = \frac{\dot b}{\dot a}\frac{a}{\dot b}\frac{\dot a}{\$}\frac{b}{\$}, 
w_{sol,2} = \frac{\dot a}{\dot b}\frac{b}{\$}, w_{sol,3} = \epsilon
\} $$

Thus, the acceptance condition of $\A$ is $\alpha = \forall x_1\exists x_2 
\exists x_3$, where $x_1$ is to be assigned a potential partial solution 
$w_{sol,k}$, and $x_2$ is to be assigned $w_{sol,k+1}$, and $x_3$ is to be 
assigned $w_{sol}$.  

During a run on a hyperword $S$ and an assignment $v:\{x_1,x_2,x_3\}\rightarrow 
S$, the NFH $\A$ checks that the upper and lower letters of $w_{sol}$ all match.
In addition, $\A$ checks that the first domino of $v(x_1)$ is indeed in $C$, and 
that $v(x_2)$ is obtained from $v(x_1)$ by removing the first tile. 
$\A$ performs the latter task by checking that the upper and lower parts of 
$v(x_2)$ are the upper and lower parts of $v(x_1)$ that have been ``shifted'' 
back appropriately. That is, if the first tile in $v(x_2)$ is the encoding of 
$[\frac{w_i}{v_i}]$, then $\A$ uses states to remember, at each point, the last 
$|w_i|$ letters of the upper part of $v(x_2)$ and the last $|v_i|$ letters of 
the lower part of $v(x_2)$, and verifies, at each point, that the next letter in 
$v(x_1)$ matches the matching letter remembered by the state.
\end{proof}

\subsection*{Theorem~\ref{thm:nfhe.nfhf.nonemptiness}}

\begin{proof}
The lower bound for both fragments follows from the \comp{NL-hardness} of the nonemptiness problem for NFA. 

We turn to the upper bound, and begin with $\nfhe$. Let $\A_\exists$ be an $\nfhe$. 
We claim that $\A_\exists$ is nonempty iff $\hat\A_\exists$ accepts some legal 
word ${\bi w}$. The first direction is trivial. For the second direction, let 
${\bi w}\in\lang{\hat\A_\exists}$, and let $S=\unzip({\bi w})$. By 
assigning $v(x_i) = {\bi w}[i]$ for every $x_i\in X$, we get $\zip(v) = {\bi w}$, and 
according to the semantics of $\exists$, we have that $\A_\exists$ accepts $S$. 
To check whether $\hat\A_\exists$ accepts a legal word, we can run a 
reachability check on-the-fly, while advancing from a letter $\sigma$ to the 
next letter $\sigma'$ only if $\sigma'$ contains $\#$ in all the positions in 
which $\sigma$ contains $\#$. 
While each transition $T = q\xrightarrow{(\sigma_1,\ldots \sigma_n)} p$ in 
$\hat\A$ is of size $k$, we can encode $T$ as a set of size $k$ of encodings of 
transitions of type $q \xrightarrow {\sigma_i} p$ with a binary encoding of 
$p,q,\sigma_i$, as well as $i,t$, where $t$ marks the index of $T$ within the 
set of transitions of $\hat\A$. 
Therefore, the reachability test can be performed within space that is logarithmic in the size of $\A_\exists$.

Now, let $\A_\forall$ be an $\nfhf$ over $X$. We claim that $\A_\forall$ is 
nonempty iff $\A_\forall$ accepts a hyperword of size $1$. 
For the first direction, let $S\in\hlang{\A_\forall}$. Then, by the semantics of 
$\forall$, we have that for every assignment $v:X\rightarrow S$, it holds that 
$\zip(v)\in\lang{\hat{\A_\forall}}$. Let $u\in S$, and let $v_u(x_i) = u$ for every 
$x_i\in X$. Then, in particular, $\zip(v_u)\in\lang{\hat{\A_\forall}}$. Then for every assignment $v:X\rightarrow \{u\}$ (which consists of the 
single assignment $v_u$), it holds that $\hat{\A_\forall}$ accepts $\zip(v)$, 
and therefore $\A_\forall$ accepts $\{u\}$. 
The second direction is trivial. 

To check whether $\A_\forall$ accepts a hyperword of size $1$, we restrict the reachability test on $\hat\A_\forall$ to  
$k$-tuples of the form $(\sigma,\sigma,\ldots \sigma)$ for $\sigma\in \Sigma$. 
\end{proof}

\subsection*{Theorem~\ref{thm:nfhef.nonemptiness}}

\begin{proof}
We begin with the upper bound.
Let $S\in\hlang{\A}$. Then, according to the semantics of the quantifiers, 
there exist $w_1,\ldots w_m \in S$, such that for every assignment 
$v:X\rightarrow S$ in which $v(x_i) = w_i$ for every $1\leq i\leq 
m$, it holds that $\hat\A$ accepts $\zip(v)$. Let $v:X\rightarrow S$ be such an 
assignment. Then, $\hat\A$ accepts 
$\zip(v_\zeta)$ for every sequence $\zeta$ of the form $(1,2,\ldots 
m,i_1,i_2,\ldots i_{k-m})$. 
In particular, it holds for such sequences in which 
$1\leq i_j\leq m$ for every $1\leq j \leq k-m$,
that is, sequences in which the last $k-m$ variables are assigned words that are assigned to the first $m$ variables. Therefore, again 
by the semantics of the quantifiers, we have that $\{v(x_1),\ldots v(x_m)\}$ is 
in $\hlang{\A}$. The second direction is trivial.

We call $\zip(v_\zeta)$ as described above a {\em witness to the nonemptiness of 
$\A$}, i.e., $\zip(v_\zeta)$ is an instantiation of the existential quantifiers.
We construct an NFA $A$ based on $\hat\A$ that is nonempty iff $\hat\A$ accepts 
a witness to the nonemptiness of $\A$.
Let $\Gamma$ be the set of all sequences of the above form.
For every sequence $\zeta = (i_1,i_2,\ldots i_k)$ in $\Gamma$, we construct 
an NFA $A_\zeta = \tuple{ \hat\Sigma,Q,Q_0,\delta_\zeta,F}$, where for every 
$q\xrightarrow{(\sigma_{i_1},\sigma_{i_2},\ldots \sigma_{i_k})} q'$ in 
$\delta$, 
we have $q\xrightarrow{(\sigma_1,\sigma_2,\ldots \sigma_k)} q'$ in 
$\delta_\zeta$.
Intuitively, $A_{\zeta}$ runs on every word ${\bi w}$ the same way that $\hat\A$ runs 
on ${\bi w}_\zeta$. Therefore, $\hat\A$ accepts a witness ${\bi w}$ to the nonemptiness of 
$\A$ iff ${\bi w}\in\lang{A_\zeta}$ for every $\zeta\in\Gamma$. 

We define $A = \bigcap_{\zeta\in\Gamma} A_\zeta$.
Then $\hat\A$ accepts a witness to the nonemptiness of $\A$ iff $A$ is 
nonempty. 
Since $|\Gamma| = m^{k-m}$, the state space of $A$ is of size $O(n^{m^{k-m}})$, where $n=|Q|$, 
and its alphabet is of size $|\hat\Sigma|$. 
Notice that for $\A$ to be nonempty, $\delta$ must be of size at least 
$|(\Sigma\cup {\#})|^{(k-m)}$, to account for all the permutations of letters in 
the words assigned to the variables under $\forall$ quantifiers (otherwise, we 
can immediately return ``empty''). Therefore, $|\hat\A|$ is $O(n\cdot 
|\Sigma|^k)$. 
We then have that the size of $A$ is $O(|\hat \A|^k)$.  
If the number $k-m$ of $\forall$ quantifiers is fixed, then $m^{k-m}$ is 
polynomial in $k$. However, now $|\hat\A|$ may be polynomial in $n,k$, and $|\Sigma|$, 
and so in this case as well, the size of $A$ is $O(|\hat A|^k)$. 

Since the nonemptiness problem for NFA is \comp{NL-complete}, the problem for $\nfhef$ 
can be decided in space of size that is polynomial in $|{\hat\A}|$. 

\bigbreak
\noindent{\bf \comp{PSPACE hardness}}
For the lower bound, we show a reduction from a polynomial version of the {\em 
corridor tiling problem}, 
defined as follows.
We are given a finite set $T$ of tiles, two relations $V \subseteq T \times T$ 
and $H \subseteq T \times T$,
an initial tile $t_0$, a final tile $t_f$, and a bound $n>0$.
We have to decide whether there is some $m>0$ and a tiling of a $n \times 
m$-grid such that
(1) The tile $t_0$ is in the bottom left corner and the tile $t_f$ is in the top 
right corner,
(2) A horizontal condition: every pair of horizontal neighbors is in $H$, and
(3) A vertical condition: every pair of vertical neighbors is in $V$.
When $n$ is given in unary notation, the problem is known to be 
\comp{PSPACE-complete}.

Given an instance $C$ of the tiling problem, we construct an $\nfhef$ $\A$ that 
is nonempty iff $C$ has a solution. 
We encode a solution to $C$ as a word $w_{sol} =w_1\cdot w_2\cdot w_m\$$ over 
$\Sigma = T\cup\{1,2,\ldots n,\$\}$, where the word $w_i$, of the form $1\cdot t_{1,i}\cdot 2 \cdot t_{2,i},\ldots n\cdot 
t_{n,i}$, describes the contents of row $i$. 

To check that $w_{sol}$ indeed encodes a solution, we need to make sure that:
\begin{enumerate}
\item $w_1$ begins with $t_0$ and $w_m$ ends with $t_f\$$.
\item $w_i$ is of the correct form.
\item Within every $w_i$, it holds that $(t_{j,i},t_{j+1,i})\in H$.
    
    \item For $w_i,w_{i+1}$, it holds that $(t_{j,i}, t_{j,i+1})\in V$ for every 
$1\leq j\leq n$.
\end{enumerate} 

Verifying items $1-3$ is easy via an NFA of size $O(n|H|)$.  
The main obstacle is item $4$. 

We describe an $\nfhef$ $\A = \tuple{T\cup \{0,1,2,\ldots n,\$\}, 
\{y_1,y_2,y_3,x_1,\ldots x_{\log(n)}\}, Q, \{q_0\},\delta, F, \alpha}$ that is 
nonempty iff there exists a word that satisfies items $1-4$.
The quantification condition $\alpha$ is $\exists y_1\exists y_2 \exists y_3\forall 
x_1 \ldots \forall x_{\log(n)}$.
The NFH $\A$ only proceeds on letters whose first three positions are of the type 
$(r,0,1)$, where $r\in T\cup\{1,\ldots n,\$\}$. Notice that this means that $\A$
requires the existence of the words $0^{|w_{sol}|}$ and $1^{|w_{sol}|}$ (the $0$ 
word and $1$ word, henceforth).
$\A$ makes sure that the word assigned to $y_1$ matches a correct solution 
w.r.t. items $1-3$ described above.  
We proceed to describe how to handle the requirement for $V$. 
We need to make sure that for every position $j$ in a row, the tile in position 
$j$ in the next row matches the current one w.r.t. $V$. We can use a state $q_j$ 
to remember the tile in position $j$, and compare it to the tile in the next 
occurrence of $j$. The problem is avoiding having to check all positions 
simultaneously, which would require exponentially many states. To this end, we 
use $\log(n)$ copies of the $0$ and $1$ words to form a binary encoding of the 
position $j$ that is to be remembered. The $\log(n)$ $\forall$ conditions make 
sure that every position within $1-n$ is checked.  

We limit the checks to words in which $x_1,\ldots x_{\log(n)}$ are the $0$ or 
$1$ words, by having $\hat\A$ accept every word in which there is a letter that 
is not over $0,1$ in positions $4,\ldots \log(n)+3$. This takes care 
of accepting all cases in which the word assigned to $y_1$ is also assigned to 
one of the $x$ variables. 

To check that $x_1,\ldots x_{\log(n)}$ are the $0$ or $1$ words, $\hat\A$ checks 
that the values in positions $4$ to $\log(n)+3$ remain constant throughout the 
run. 
In these cases, upon reading the first letter, $\hat\A$ remembers the value $j$ 
that is encoded by the constant assignments to $x_1,\ldots x_{\log(n)}$ in a 
state, and makes sure that throughout the run, the tile that occurs in the 
assignment $y_1$ in position $j$ in the current row matches the tile in position 
$j$ in the next row. 

We construct a similar reduction for the case that the number of $\forall$ 
quantifiers is fixed: instead of encoding the position by $\log(n)$ bits, we can 
directly specify the position by a word of the form $j^*$, for every $1\leq 
j\leq n$. 
Accordingly, we construct an $\nfhef$ over $\{x, y_1,\ldots y_{n},z\}$, with a 
quantification condition $\alpha = \exists x\exists y_1 \ldots \exists 
y_{n}\forall z$. The NFA $\hat\A$ advances only on letters whose assignments to 
$y_1,\ldots y_n$ are always $1,2,\ldots n$, respectively, and checks only words 
assigned to $z$ that are some constant $1\leq j\leq n$. Notice that the fixed 
assignments to the $y$ variables leads to $\delta$ of polynomial size.  
In a hyperword accepted by $\A$, the word assigned to $x$ is $w_{sol}$, and the 
word assigned to $z$ specifies which index should be checked for conforming to 
$V$.
\end{proof}

\subsection*{Theorem~\ref{thm:nfh.membership.finite}}

\begin{proof}
We can decide the membership of $S$ in $\hlang{\A}$ by iterating over all relevant assignments from $X$ to $S$, and for every such assignment $v$, checking on-the-fly whether $\zip(v)$ is accepted by $\hat\A$. 
This algorithm uses space of size that is polynomial in $k$ and logarithmic in $|\A|$ and in $|S|$. 

In the case that $k' = O(\log k)$, an \comp{NP} upper bound is met by iterating over 
all assignments to the variables under $\forall$, while guessing assignments to 
the variables under $\exists$. For each such assignment $v$, checking whether $\zip(v)\in\lang{\hat\A}$ can be done on-the-fly. 

We show \comp{NP-hardness} for this case by a reduction from the Hamiltonian cycle problem. 
Given a graph $G = \tuple{V,E}$ where $V = \{v_1,v_2,\ldots, v_n\}$ and 
$|E|=m$, 
we construct an $\nfhe$ $\A$ over $\{0,1\}$ with $n$ states, $n$ variables, 
$\delta$ of size $m$, and a hyperword $S$ of size $n$, as follows. $S = 
\{w_1,\ldots, w_n\}$, where $w_i$ is the word over $\{0,1\}$ in which all 
letters 
are $0$ except for the $i$'th. 
The structure of $\hat\A$ is identical to that of $G$, and we set $Q_0 = F = 
\{v_1\}$. For the transition relation, for every $(v_i,v_j)\in E$, we have 
$(v_i, \sigma_i,v_j)\in \delta$, where $\sigma_i$ is the letter over $\{0,1\}^n$ 
in which all positions are $0$ except for position $i$. 
Intuitively, the $i$'th letter in an accepting run of $\hat\A$ marks traversing 
$v_i$. Assigning $w_j$ to $x_i$ means that the $j$'th step of the run 
traverses $v_i$. Since the words in $w$ make sure that every $v\in V$ is 
traversed exactly once, and that the run on them is of length $n$, we have that 
$\A$ accepts $S$ iff there exists some permutation $p$ of the words in $S$ such 
that $p$ matches a Hamiltonian cycle in $G$.

\noindent{\it remark}
To account for all the assignments to the $\forall$ variables, $\delta$ -- and therefore, $\hat\A$ -- must be of size at least $2^{k'}$ (otherwise, we can return ``no'').
We then have that if $k = O(k')$, then space of size $k$ is logarithmic in $|\hat\A|$, and so the problem in this case can be solved within logarithmic space.
A matching NL lower bound follows from the membership problem for NFA. 
\end{proof}

\subsection*{Theorem~\ref{thrm:membershipFULL}}

\begin{proof}
Let $A = \tuple{\Sigma, P, P_0, \rho,F}$ be an NFA, and let $\A = \tuple{\Sigma, 
\{x_1,\ldots, x_k\}, Q, Q_0, \delta, {\cal F},\alpha}$ be an NFH. 

First, we construct an NFA $A'=\tuple{\Sigma\cup\{\#\}, P', P'_0, \rho', F'}$
by extending the alphabet of $A$ to $\Sigma\cup\{\#\}$, adding a new and 
accepting state $p_f$ to $P$ with a self-loop labeled by $\#$, and transitions 
labeled by $\#$ from every $q\in F$ to $p_f$. 
The language of $A'$ is then $\lang{A}\cdot \#^*$.  
We describe a recursive procedure (iterating over $\alpha$) for deciding 
whether $\lang{A}\in\hlang{\A}$.

For the case that $k=1$, it is easy to see that if $\alpha = \exists x_1$, then 
$\lang{A}\in\hlang{\A}$ iff $\lang{A}\cap \lang{\hat{\A}} \neq \emptyset$.
Otherwise, if $\alpha = \forall x_1$, then $\lang{A}\in\hlang{\A}$ iff 
$\lang{A}\notin \hlang{\overline{\A}}$, where $\overline{\A}$ is the NFH for 
$\overline{\hlang{\A}}$ described in Theorem~\ref{thm:nfh.operations}. 
Notice that the quantification condition for $\overline{\A}$
is $\exists x_1$, and so this conforms to the base case.

For $k>1$, we construct a sequence of NFH $\A_1, \A_2, \ldots, \A_k$. 
If $\quant_1 = \exists$ then we set $\A_1 = \A$, and otherwise we set $\A_1 = \overline{\A}$.
Let $\A_i = \tuple{\Sigma, \{x_i,\ldots, x_k\}, Q_i, 
Q^0_i, \delta_i, {\cal F}_i,\alpha_i}$.
If $\alpha_i$ starts with $\exists$, then 
we construct $\A_{i+1}$ as follows. 

The set of variables of $\A_{i+1}$ is $\{x_{i+1},\ldots, x_k\}$, and the 
quantification condition $\alpha_{i+1}$ is \linebreak $\quant_{i+1}x_{i+1}\cdots 
\quant_kx_k$, for $\alpha_i = \quant_ix_i \quant_{i+1}\cdots \quant_kx_k$. 
The set of states of $\A_{i+1}$ is $ Q_i\times P'$, and the set of initial 
states is $Q_i^0\times P_0$. The set of accepting states is ${\cal F}_i\times 
F'$.
For every 
$(q\xrightarrow{(\sigma_i,\ldots,\sigma_k)}q')\in\delta_i$ and every 
$(p\xrightarrow{\sigma_i}p')\in \rho$, we have 
$((q,p)\xrightarrow{(\sigma_{i+1},\ldots, \sigma_k)}(q',p'))\in\delta_{i+1}$. 
Then, $\hat\A_{i+1}$ accepts a word $\zip(u_1,u_2,\ldots, u_{k-i})$ iff there 
exists a word $u\in \lang{A}$, such that $\hat\A_{i}$ accepts 
$\zip(u,u_1,u_2,\ldots, u_{k-i})$. 

 Let $v:\{x_{i},\ldots, x_k\}\rightarrow \lang{A}$.
Then $\lang{A}\models _v (\alpha_i,\A_i)$ iff there exists $w\in \lang{A}$ such 
that $\lang{A}\models_{v[x_i\rightarrow w]} (\alpha_{i+1},\A_i)$.
For an assignment $v':\{x_{i+1},\ldots, x_k\}\rightarrow \lang{A}$, it holds 
that 
$\zip(v')$ is accepted by $\hat{\A}_{i+1}$ iff there exists a word $w\in 
\lang{A}$ such that $\zip(v)\in\lang{\hat{\A}_i}$, where $v$ is obtained from 
$v'$ 
by setting $v(x_i) = w$. 
Therefore, we have that $\lang{A}\models_{v[x_i\rightarrow 
w]}(\alpha_{i},\A_i)$ 
iff $\lang{A}\models_{v'} (\alpha_{i+1}, \A_{i+1})$, that is, 
$\lang{A}\in\hlang{\A_i}$ iff $\lang{A}\in\hlang{\A_{i+1}}$.

If $\alpha_i$ starts with $\forall$, then we have that $\lang{A}\in \hlang{\A_i}$ iff 
$\lang{A}\notin \overline{\hlang{\A_i}}$. We construct $\overline{\A_{i}}$ for 
$\overline{\hlang{\A_i}}$ as described in Theorem~\ref{thm:nfh.operations}. The 
quantification condition of $\overline{\A_{i}}$ then begins with $\exists x_i$, and we 
apply the previous case, and construct $\A_{i+1}$ w.r.t. $\overline{\A_{i}}$, to check 
for non-membership.

Every $\forall$ quantifier requires complementation, which is exponential in 
$n$, the number of states in $\A$. Therefore, in the worst case, the complexity of this algorithm is 
$O(2^{2^{...^{|Q||A|}}})$, where the tower is of height $k$. If the number of 
$\forall$ quantifiers is fixed, then the complexity is $O(|Q||A|^k)$. 
\end{proof}

\subsection*{Theorem~\ref{thrm:containment}}

\begin{proof}
For the lower bound, we show a reduction from the containment problem for NFA, which is known to be \comp{PSPACE-hard}.
Let $A_1,A_2$ be NFA. We ``convert'' them to NFH $\A_1,\A_2$ by adding to both a single variable $x$, and a quantification condition $\forall x$.
By the semantics of the $\forall$ quantifier, we have that $\hlang{\A_1} = \{S | S\subseteq \lang{A_1}\}$, and similarly for $\A_2$. 
Therefore, we have that $\hlang{\A_1}\subseteq \hlang{\A_2}$ iff $\lang{A_1}\subseteq \lang{A_2}$.

For the upper bound, first notice that complementing an $\nfhf$ yields an 
$\nfhe$, and vice versa. 
Consider two NFH $\A_1$ and $\A_2$.
Then $\hlang{\A_1}\subseteq\hlang{\A_2}$ iff 
$\hlang{\A_1}\cap\overline{\hlang{\A_2}}  =  \emptyset$. 
We can use the constructions in the proof of Theorem~\ref{thm:nfh.operations} 
to compute a matching NFH $\A = 
\A_1\cap\overline{\A_2}$, and check its nonemptiness.
The complementation construction is exponential in $n_2$, the number of states 
of $\A_2$, and the intersection construction is polynomial in $|\A_1|, 
|\overline{\A_2}|$. 

If $\A_1\in\nfhe$ and $\A_2\in\nfhf$ or vice versa, then $\A$ is an $\nfhe$ or 
$\nfhf$, respectively, whose nonemptiness can be decided in space that is 
logarithmic in $|\A|$. 

Now, consider the case where $\A_1$ and $\A_2$ are both $\nfhe$ or both $\nfhf$.
It follows from the proof of Theorem~\ref{thm:nfh.operations}, that for two 
NFH $\A,\A'$, the quantification condition of $\A\cap\A'$ may be any 
interleaving of the quantification conditions of $\A$ and $\A'$. Therefore, if 
$\A_1,\A_2\in\nfhe$ or $\A_1,\A_2\in\nfhf$, we can construct $\A$ to be an 
$\nfhef$. 
This is also the case when $\A_1\in\nfhef$ and $\A_2\in\nfhe$ or $\A_2\in\nfhf$.

Either $\A_2$ or $\overline{\A_2}$ is an $\nfhf$, whose underlying NFA has a transition relation of size that is exponential in $k$ (otherwise the $\nfhf$ is empty). The same holds for $\A_1\in\nfhef$. The PSPACE upper bound of Theorem~\ref{thm:nfhef.nonemptiness} is derived from the number of variables and not from the state-space of the NFH. Therefore, while $|\bar{\A_2}|$ is exponential in the number of states of $\A_2$, checking the nonemptiness of $\A$ is in \comp{PSPACE}. 
\end{proof}

\subsection*{Lemma~\ref{lem:langeq}}

\begin{proof}
We begin with $\nfhf$.
For the first direction, since $\lang{\hat\A_\forall'}\subseteq \lang{\hat\A_\forall}$, we have 
$\hlang{\A_\forall'}\subseteq \hlang{\A_\forall}$.
For the second direction, let 
$S\in\hlang{A_\forall}$. Then for every $v:S\rightarrow X$, it holds that 
$\zip(v)\in\lang{\hat\A_\forall}$. Also,   $\zip(v')\in\lang{\hat\A_\forall}$ for every sequence 
$v'$ of $v$. Then $\zip(v)$ and all its sequences are in $\lang{\hat\A_\forall'}$. Since 
this holds for every $v:X\rightarrow S$, we have that $S\in\hlang{\A_\forall'}$.

We proceed to $\nfhe$.
For the first direction, since $\lang{\hat\A_\exists}\subseteq \lang{\hat\A_\exists'}$, we have 
$\hlang{\A}\subseteq \hlang{\A'}$.
For the second direction, let 
$S\in\hlang{A_\exists'}$. Then there exists $v:S\rightarrow X$ such that 
$\zip(v)\in\lang{\hat\A_\exists'}$. Then
$\zip(v)$ is a permutation of some word $\zip(v')\in\lang{\hat\A_\exists}$.  
According to the semantics of the $\exists$ quantifier, we have that $S\in \hlang{\A_\exists}$.
\end{proof}

\subsection*{Lemma~\ref{lem:permutation.sequence.complete}}

\begin{proof}
We begin with $\nfhf$.
To construct $\A_\forall'$ given $\A_\forall$, we use a similar construction to the one 
presented in the proof of Theorem~\ref{thm:nfhef.nonemptiness}. Essentially, for 
every sequence $\zeta$ of $(1,2,\ldots, k)$, we construct an NFA $A_\zeta$, in 
which every run on a word ${\bi w}$ matches a run of $\hat\A_\forall$ on ${\bi w}_\zeta$. 
The $\nfhf$ $\A'$ is then obtained from $\A_\forall$ by replacing the underlying NFA 
with $\bigcap_{\zeta\in\Gamma} A_\zeta$, where $\Gamma$ is the set of sequences 
of $(1,2,\ldots, k)$.

For $\nfhe$, similarly to the case of $\nfhf$, we construct $\A_\exists'$ given $\A_\exists$ by constructing 
$A_\zeta$ for every permutation $\zeta$ of $(1,2,\ldots, k)$.
In this case, the $\nfhe$ $\A_\exists'$ is obtained from $\A_\exists$ by replacing the 
underlying NFA with $\bigcup_{\zeta\in\Gamma} A_\zeta$, where $\Gamma$ is the 
set of permutations of $(1,2,\ldots, k)$.
\end{proof}

\subsection*{Theorem~\ref{thm:permutation.sequence.complete}}

\begin{proof}
We begin with $\nfhf$.
For the first direction, let ${\bi w}\in\lang{\hat\A_\forall}$.
Since $\A_1$ is sequence-complete, then ${\bi w}'\in\lang{\hat\A_1}$ for every 
sequence ${\bi w}'$ of ${\bi w}$. Then, by the semantics of the $\forall$ quantifier, we 
have that $\unzip({\bi w})\in\hlang{\A_1}$. Therefore, $\unzip({\bi w})\in\hlang{\A_2}$, and so 
${\bi w}$ (and all its sequences) are in $\lang{\hat\A_2}$. A similar argument can be 
made to show that for every ${\bi w}\in\lang{\hat{\A_2}}$, it holds that 
${\bi w}\in\lang{\hat\A_1}$. Therefore, $\lang{\hat\A_1} = \lang{\hat \A_2}$.
The second direction is trivial.

We continue to $\nfhe$.
For the first direction, let ${\bi w}\in\lang{\hat\A_1}$. Then 
$\unzip({\bi w})\in\lang{\A_1}$. 
Then, by the semantics of the $\exists$ quantifier, there exists some 
permutation ${\bi w}'$ of ${\bi w}$ such that ${\bi w}'\in\lang{\hat\A_2}$.
Since $\A_2$ is permutation-complete, we have that ${\bi w}\in\lang{\hat\A_2}$. A 
similar argument can be made to show that for every ${\bi w}\in\lang{\hat{\A_2}}$, it 
holds that ${\bi w}\in\lang{\hat\A_1}$. Therefore, $\lang{\hat\A_1} = \lang{\hat 
\A_2}$.
The second direction is trivial.
\end{proof}

\end{document}